\documentclass[a4paper,UKenglish,autoref]{lipics-v2019}

\usepackage{tcolorbox}
\usepackage{bm,tikz}
\usepackage[usestackEOL]{stackengine}

\usepackage{algorithm}
\usepackage[noend]{algpseudocode}

\usepackage{comment}
\usetikzlibrary{decorations.markings}
\usepackage{xspace}

\usepackage{multirow}

\newcommand{\ignore}[1]{}

\newcommand{\bS}{\ensuremath{\bm{S}}}

\newcommand{\bP}{\ensuremath{\bm{P}}}
\newcommand{\bPP}{\ensuremath{\bm{\mathcal{P}}}}
\newcommand{\bI}{\ensuremath{\bm{I}}}
\newcommand{\bT}{\ensuremath{\bm{\theta}}}

\newcommand{\bX}{\ensuremath{\bm{X}}}

\newcommand{\bA}{\ensuremath{\bm{A}}}
\newcommand{\bB}{\ensuremath{\bm{B}}}

\newcommand{\supp}[1]{\ensuremath{\textsc{supp}(#1)}}

\newcommand{\prot}{\ensuremath{\Pi}}

\renewcommand{\SS}{\ensuremath{\mathcal{S}}}

\newcommand{\protds}{\ensuremath{\Pi}_{\textsf{DS}}}

\newcommand{\bSS}{\ensuremath{\bm{\SS}}}

\newcommand{\bR}{\ensuremath{\bm{R}}}

\newcommand{\Paren}[1]{\Big(#1\Big)}
\newcommand{\Bracket}[1]{\Big[#1\Big]}

\newcommand{\ICost}[2]{\ensuremath{\textnormal{\textsf{ICost}}_{#2}(#1)}\xspace}
\newcommand{\IC}[3]{\ensuremath{\textnormal{\textsf{IC}}_{#2}^{#3}(#1)}\xspace}
\newcommand{\CC}[3]{\ensuremath{\textnormal{\textsf{CC}}_{#2}^{#3}(#1)}\xspace}

\newcommand{\bprotds}{\ensuremath{\bm{\protds}}}
\newcommand{\bprot}{\ensuremath{\bm{\prot}}}

\newcommand{\diste}{\ensuremath{\dist_{\textsf{est}}}}

\newenvironment{tbox}{\begin{tcolorbox}[
		enlarge top by=5pt,
		enlarge bottom by=5pt,
		 boxsep=0pt,
                  left=4pt,
                  right=4pt,
                  top=10pt,
                  arc=0pt,
                  boxrule=1pt,toprule=1pt,
                  colback=white
                  ]%%
	}
{\end{tcolorbox}}

\newcommand{\toShrink}{-.25cm}
\newcommand{\toShrinkEnu}{-.2cm}
\newcommand{\toShrinkEqn}{-.45cm}

\renewcommand{\bar}[1]{\overline{#1}}

\newcommand{\card}[1]{\left\vert{#1}\right\vert}
\newcommand{\opt}{\text{{OPT}}\xspace}

\newcommand{\set}[1]{\ensuremath{\left\{ #1 \right\}}}

\newcommand{\poly}{\mbox{\rm poly}}

\newcommand{\REM}[1]{}
\newcommand{\Ot}{\widetilde{O}}

\newcommand{\union}{\ensuremath{\cup}}

\newcommand{\textbox}[2]{
{
\begin{tbox}
\textbf{#1}
{#2}
\end{tbox}
}
}

\newcommand{\eat}[1]{}

\DeclareMathOperator*{\Exp}{\ensuremath{\textnormal{E}}}
\DeclareMathOperator*{\Prob}{\ensuremath{\textnormal{Pr}}}
\renewcommand{\Pr}{\Prob}
\newcommand{\Ex}{\Exp}

\newcommand{\norm}[1]{\ensuremath{\|#1\|}}

\newcommand{\PR}[1]{\ensuremath{\Pr\left(#1\right)}\xspace}

\newcommand{\paren}[1]{\ensuremath{\left(#1\right)}\xspace}

\newcommand{\dist}{\ensuremath{\mathcal{D}}}

\newcommand{\Salpha}{\ensuremath{\widehat{\mathcal{S}}}}

\newcommand{\FC}{\ensuremath{\mathcal{F}}}

\newcommand{\dsest}{\textnormal{\ensuremath{\textsf{DomSet}_{\textsf{est}}}}\xspace}

\newcommand{\istar}{\ensuremath{{i^*}}\xspace}

\algnewcommand{\IfThenElse}[3]{% \IfThenElse{<if>}{<then>}{<else>}
  \State \algorithmicif\ #1\ \algorithmicthen\ #2\ \algorithmicelse\ #3}

\tikzset{
  hello/.style={
    label={[draw,circle,xscale={1.25},minimum size=14mm+#1*21mm]center:},
    label={[red,yshift=#1*11mm]center:hello}
  }
}

\title{\mbox{Towards a Theory of Parameterized Streaming} Algorithms \footnote{A preliminary version of this paper will appear in IPEC 2019. Supported by ERC grant 2014-CoG 647557.}
}

\titlerunning{Towards a Theory of Parameterized Streaming Algorithms}
\author{Rajesh Chitnis\footnote{Work done while at University of Warwick, UK}}{School of Computer Science, University of Birmingham, UK.}{rajeshchitnis@gmail.com}{}{}
\author{Graham Cormode}{University of Warwick, UK.}{g.cormode@warwick.ac.uk}{}{}
\authorrunning{Rajesh Chitnis and Graham Cormode}
\Copyright{Rajesh Chitnis and Graham Cormode}
% mandatory: Please provide 1-5 keywords
\EventEditors{}
\EventNoEds{0}
\EventLongTitle{}
\EventShortTitle{}
\EventAcronym{}
\EventYear{}
\EventDate{}
\EventLocation{}
\EventLogo{eatcs}
\SeriesVolume{}
\ArticleNo{}
\newcommand{\Pe}{\texttt{P}\xspace}
\newcommand{\NP}{\texttt{NP}\xspace}
\newcommand{\FPT}{\texttt{FPT}\xspace}

\newcommand{\fps}{\texttt{FPS}\xspace}
\newcommand{\subps}{\texttt{SubPS}\xspace}
\newcommand{\semips}{\texttt{SemiPS}\xspace}
\newcommand{\supps}{\texttt{SupPS}\xspace}
\newcommand{\bps}{\texttt{BrutePS}\xspace}

\newcommand{\dict}{\texttt{Dict}\xspace}

\newcommand{\nextt}{\texttt{Next}\xspace}

\newcommand{\T}{\ensuremath{T}}
\newcommand{\tw}{\ensuremath{\textbf{tw}}}

\acknowledgements{We thank MohammadTaghi Hajiaghayi, Robert Krauthgamer and Morteza Monemizadeh for helpful discussions. Algorithm~\ref{alg:streaming-branching} was suggested to us by Arnold Filtser.}

\begin{document}
%has to be here, or it won't work:
\renewcommand*{\sectionautorefname}{Section}
\renewcommand*{\subsectionautorefname}{Section}
\renewcommand*{\subsubsectionautorefname}{Section}

\strutlongstacks{T}
\nolinenumbers

% \begin{titlepage}

\maketitle

\begin{abstract}
%Instead of measuring the running time as a function of the input size
%only,  parameterized complexity attempts to give a more fine-grained analysis of the complexity of problems by expressing the costs in terms of additional
%parameters.
Parameterized complexity attempts to give a more fine-grained analysis of the complexity of problems: instead of measuring the running time as a function of only the input size, we analyze the running time with respect to additional parameters.
This approach has proven to be highly successful in
delineating our understanding of \NP-hard problems.
%coping with NP-hardness.
Given this success with the TIME resource, it seems but natural to use
this approach for dealing with the SPACE resource.
First attempts in this direction have considered a few individual problems,
with some success:
Fafianie and Kratsch [MFCS'14] and Chitnis et al. [SODA'15]
introduced the notions of streaming kernels and parameterized streaming algorithms respectively.
For example, the latter shows how to refine the $\Omega(n^2)$ bit lower bound for finding a minimum Vertex Cover (VC) in the streaming setting by designing an algorithm for the parameterized $k$-VC problem which uses $O(k^{2}\log n)$ bits.

In this paper, we initiate a systematic study of graph problems from
the paradigm of parameterized streaming algorithms.
We first define a natural hierarchy of space complexity classes
of \fps, \subps, \semips, \supps and \bps, and then obtain tight
classifications for several well-studied graph problems such as
Longest Path, Feedback Vertex Set, Dominating Set, Girth, Treewidth, etc. into this hierarchy (see Figure~\ref{landscape} and Figure~\ref{table}).
On the algorithmic side, our parameterized streaming algorithms use
techniques from the FPT world such as bidimensionality, iterative compression and
bounded-depth search trees. On the hardness side, we obtain lower bounds for the parameterized streaming complexity of various problems via novel reductions from problems in communication complexity.
We also show a general (unconditional) lower bound for space complexity of parameterized streaming algorithms for a large class of problems inspired by the recently developed frameworks for showing (conditional) kernelization lower bounds.

%how (conditional) lower bounds for kernels and W-hard problems translate to lower bounds for parameterized streaming algorithms.

Parameterized algorithms and streaming algorithms are approaches to
cope with TIME and SPACE intractability respectively. It is our hope
that this work on parameterized streaming algorithms leads to two-way
flow of ideas between these two previously separated areas of theoretical computer science.

%\todo{Things to be done below}
%\begin{itemize}
%  \item AND-OR inspiration
%  \item $k$-OCT lower bound? For $k=0$ we have $\Theta(n\cdot \log n)$ bound but what happens for $k=1$?
%  \item Say $\Omega(n)$ holds for multipass for k-Path, k-FVS and k-TW and $\Omega(n\log n)$ holds even for 1-pass. Both hold for lower bounds. Add Guruswami-Onak bounds too.
%  \item Two DISC papers give a lower bound on number of passes for $k$-VC. Maybe we can use bit-gadgets for other problems?
%  \item Directed problems are harder usually: many problems like acyclicity, $k$-path, $s\leadsto t$ reachability,directed OCT need $\Omega(n^2)$ for constant $k$. Maybe this is because directions allow us to force structure in lower bounds. Are there any more directed problems?
%  \item No (3/2) approx for k-VC using $o(k)$ space and no (2)-approx for k-OCT using $o(k)$ space from BHH
%\end{itemize}

\end{abstract}

%\newpage
%\setcounter{page}{1}

\section{Introduction}

Designing and implementing efficient algorithms is at the heart of computer science.
Traditionally, efficiency of algorithms has been measured with respect to
running time as a function of instance size.
From this perspective, algorithms are said to be efficient if they can
be solved in time which is bounded by some polynomial function of the
input size.
However, very many interesting problems are \NP-complete, and so are
grouped together as ``not known to be efficient''.
This fails to discriminate within a large heterogenous group of
problems, and in response
the theory of \emph{parameterized (time) algorithms} was developed in
late 90's by Downey and Fellows~\cite{df99}.
Parameterized complexity attempts to delineate the complexity
of problems by expressing the costs in terms of additional
parameters.
%algorithms and complexity is essentially a two-dimensional analogue of ``\Pe vs \NP". The running time is analyzed in finer detail: instead of expressing it as a function of only the input size $n$, one
%or more parameters of the input instance are defined, and we investigate the effect of these parameters on the running time.
%The goal is to design algorithms that work efficiently if the parameters of the input instance are small, even if the size of
%the input is large.
Formally, we say that a problem is \emph{fixed-parameter tractable}
(\FPT) with respect to parameter $k$ if the problem can be solved in
time $f(k)\cdot n^{O(1)}$ where $f$ is a computable function and $n$
is the input size.
For example,
%a folklore algorithm based on search trees shows that
the problem of checking if a graph on $n$ vertices has a vertex cover
of size at most $k$ can be solved in $2^k \cdot n^{O(1)}$ time.
%The goal of parameterized algorithms is to capture the combinatorial
%explosion into the parameters, rather than the input size which can
%be significantly larger.
The study of various parameters helps to understand which parameters
make the problem easier (\FPT) and which ones cause it to be hard.
The parameterized approach towards \NP-complete problems has led to
development of various algorithmic tools such as kernelization,
iterative compression, color coding, and more~\cite{df13,pc-book}.
%This area of research (also sometimes known as multivariate algorithms and fine-grained algorithms) is one of the most active areas of theoretical computer science with hundreds of new results each year. We refer the reader to~ for more details.

%However, around half a century ago, researchers discovered the phenomenon of \NP-completeness~\cite{cook-np-completeness,karp-np-completeness}. Since then, tens of thousands of problems have been shown to be \NP-complete. This includes most of the interesting algorithmic problems one needs to solve in real-world applications. The existence of a polynomial time algorithm for any \NP-complete problem immediately implies existence of polynomial time algorithms for \emph{all} \NP-complete problems. However, despite significant efforts by the algorithms community, no \NP-complete problem has yet succumbed to a polynomial time algorithm. This has led to the widely believed assumption of \Pe$\neq$\NP, i.e., no \NP-complete problem has a polynomial time algorithm. In other words, in order to solve \NP-complete problems exactly one has to accept superpolynomial (and often exponential) time algorithms.
%One of the biggest open problem in computer since whether P

%\paragraph{Parameterized Algorithms:}

\noindent \textbf{Kernelization:} %\todo{How much detail to write depends on how much we use it in the paper}
A key concept in fixed parameter tractability is that of kernelization which is an efficient preprocessing algorithm to produce a
smaller, equivalent output called the ``kernel''.
Formally, a kernelization algorithm for a parameterized problem $Q$ %\subseteq \Sigma^{*}\times \mathbb{N}$
is an algorithm which takes as an instance $\langle x, k\rangle$ and
outputs in time polynomial in $(|x| + k)$ an equivalent\footnote{By
  equivalent we mean that $\langle x, k \rangle \in Q \Leftrightarrow
  \langle x', k' \rangle \in Q$} instance $\langle x', k' \rangle$
such that $\max\{|x'|, k'\}\leq f(k)$ for some computable function
$f$. The output instance $\langle x', k'\rangle$ is called the kernel,
while the function $f$ determines the size of the kernel.
Kernelizability is equivalent to fixed-parameter tractability, and
designing compact kernels is an important question.
In recent years, (conditional) lower bounds on kernels have emerged~\cite{DBLP:journals/jcss/BodlaenderDFH09,dell-and,dell,DBLP:conf/focs/Drucker12,fortnow-santhanam}.

%A folklore observation is that a (decidable) parameterized problem $Q$ has an \FPT algorithm if and only if it has a kernel. Kernelization is an active subarea of parameterized algorithms. After proving that a problem is FPT, the next natural step is to try to design a kernel (the smaller the better). In general, the size of the kernel can be exponential in $k$, but ideally we want the kernel size to be polynomial in $k$.Some problems such as $k$-FVS and $k$-VC admit polynomial kernels. On the other hand, some lower bound techniques have been developed over the last few years to show that certain problems do not admit polynomial kernels unless $\NP\subseteq \text{co}\NP/\text{poly}$. We refer the reader to the surveys~\cite{DBLP:conf/birthday/LokshtanovMS12,DBLP:journals/disopt/MisraRS11,DBLP:journals/eatcs/Kratsch14} for more information on kernelization.

\noindent \textbf{Streaming Algorithms:}
A very different paradigm for handling large problem instances arises
in the form of streaming algorithms.
The model is motivated by sources of data arising in communication
networks and activity streams that are considered to be too big to store conveniently.
This places a greater emphasis on the space complexity of algorithms.
A streaming algorithm processes the input in one or a few read-only
passes, with primary focus on the storage space needed.
In this paper we consider streaming algorithms for graph problems over
fixed vertex sets, where information about the edges arrives edge by edge~\cite{henzinger1998computing}.
We consider variants where edges can be both inserted and deleted, or
only insertions are allowed.
We primarily consider single pass streams, but also give some multi-pass
results.

\subsection{Parameterized Streaming Algorithms and Kernels}

Given that parameterized algorithms have been extremely successful for
the TIME resource, it seems natural to also use it attack the SPACE
resource.
In this paper, we advance the model of parameterized streaming
algorithms, and start to flesh out a hierarchy of complexity classes.
We focus our attention on graph problems, by analogy with FPT, where
the majority of results have addressed graphs.
From a space perspective, there is perhaps less headroom than when
considering the time cost:
for graphs on $n$ vertices, the entire graph can be stored using
$O(n^2)$ space\footnote{Throughout the paper, by space we mean words/edges/vertices. Each word can be represented using
  $O(\log n)$ bits}.
Nevertheless, given that storing the full graph can be prohibitive,
there are natural space complexity classes to consider.
We formalize these below, but informally, the classes partition the
dependence on $n$ as:
(i) (virtually) independent of $n$; (ii) sublinear in $n$; (iii) (quasi)linear in
$n$;
(iv) superlinear but subquadratic in $n$;
and (v) quadratic in $n$.

Naively, several graph problems have strong lower bounds: for
example, the problem of finding a minimum vertex cover on graphs of
$n$ vertices has a lower bound of $\Omega(n^2)$ bits.
However, when we adopt the parameterized view, we seek
streaming algorithms for (parameterized) graph problems whose space
can be expressed as a function of
\emph{both} the number of vertices $n$ and the parameter $k$.
With this relaxation, we can separate out the problem space and start
to populate our hierarchy.
We next spell out our results, which derive from a variety of upper
and lower bounds building on the streaming and FPT literature.

\subsection{Our Results \& Organization of the paper}

For a graph problem with parameter $k$, there can be several possible choices for the space complexity needed to solve it in the streaming setting. In this paper, we first define some natural space complexity classes below:
%\todo{Should we state everything in terms of bits, or in terms of words?}
%\footnote{Indeed, this list is by no means exhaustive}:
\begin{enumerate}
  \item \underline{\textit{$\Ot(f(k))$ space:}} Due to the connection to running time of \FPT algorithms, we call the class of parameterized problems solvable using $\Ot(f(k))$ bits as \fps (fixed-parameterized streaming)\footnote{Throughout this paper, we use the $\tilde{O}$ notation to hide $\log^{O(1)} n$ factors}.

  \item \underline{\textit{Sublinear space:}} When the dependence on $n$ is sublinear, we call the class of parameterized problems solvable using $\Ot(f(k)\cdot n^{1-\epsilon})$ bits as \subps (sublinear  parameterized streaming)

  \item \underline{\textit{Quasi-linear space:}} Due to the connection to the semi-streaming model~\cite{andrew-semi-streaming,muthu-book}, we call the set of problems solvable using $\Ot(f(k)\cdot n)$ bits as \semips (parameterized semi-streaming).

  \item \underline{\textit{Superlinear, subquadratic space:}} When the dependence
    on $n$ is superlinear (but subquadratic), we call the class of
    parameterized problems solvable using $\Ot(f(k)\cdot
    n^{1+\epsilon})$ bits (for some $1>\epsilon>0$) as \supps (superlinear parameterized streaming).

  \item \underline{\textit{Quadratic space:}}
  %Any graph problem can be solved using $O(n^2)$ bits by storing the adjacency matrix. Hence,
    We call the set of graph problems solvable using $O(n^{2})$ bits as \bps (brute-force parameterized streaming). Note that every graph problem is in \bps since we can just store the entire adjacency matrix using $O(n^2)$ bits (see Remark~\ref{rem:unbounded-computation}).
\end{enumerate}

\begin{remark}
Formally, we need to consider the following 7-tuple when we attempt to find its correct position in the aforementioned hierarchy of complexity classes:
$$ [\text{Problem, Parameter, Space,} \#\ \text{of Passes, Type of Algorithm, Approx. Ratio, Type of Stream}] $$
By type of algorithm, we mean that the algorithm could be deterministic or randomized. For the type of stream, the standard alternatives are (adversarial) insertion, (adversarial) insertion-deletion, random order, etc. Figure~\ref{vc-case-study} gives a list of results for the $k$-VC problem (as a case study) in various different settings. Unless stated otherwise, throughout this paper, we consider the space requirement for 1-pass exact deterministic algorithms for problems with the standard parameter (size of the solution) in insertion-only streams.
\label{rem:7-tuple}
\end{remark}

\begin{remark}
There are various different models for streaming algorithms depending on how much computation is allowed on the stored data. In this paper, we consider the most general model by allowing \emph{unbounded computation} at each edge update, and also at the end of the stream.
%This strengthens our lower bounds. All of our (and previously known) parameterized streaming algorithms use FPT computation (or even polynomial) time at each edge update or at the end of the stream.
\label{rem:unbounded-computation}
\end{remark}

Our goal is to provide a tight classification of graph problems into
the aforementioned complexity classes.
We make  progress towards this goal as follows:
Section~\ref{sec:algo-fpt-to-streaming} shows how various techniques from the
FPT world such as iterative compression, branching, bidimensionality, etc. can also be used to design parameterized streaming
algorithms. First we investigate whether one can further improve upon the \fps algorithm of Chitnis et al.~\cite{rajesh-soda-15} for $k$-VC which uses $O(k^{2}\cdot \log n)$ bits and one pass. We design two algorithms for $k$-VC which use $O(k\cdot \log n)$ bits\footnote{Which is essentially optimal since the algorithm also returns a VC of size $k$ (if one exists)}: an $2^{k}$-pass algorithm using bounded-depth search trees (Section~\ref{subsec:branching}) and an $(k\cdot 2^{2k})$-pass algorithm using iterative compression (Section~\ref{subsec:ic}).
%In Section~\ref{subsec:branching} we use the technique of bounded-depth search trees (also known as branching) to obtain an algorithm (Algorithm~\ref{alg:streaming-branching}) which uses $O(k\log n)$ bits\footnote{Which is essentially optimal since the algorithm also returns a VC of size $k$ (if one exists)} and $2^{k}$ passes. Note that this algorithm has optimal storage, and number of passes is independent of the input size $n$. In Section~\ref{subsec:ic} we use the technique of iterative compression to design an \fps algorithm (Algorithm~\ref{alg:streaming-ic}) for $k$-VC which uses $O(k)$
%space and $O(2^{k}\cdot (n-k))$ passes.
Finally, Section~\ref{subsec:bidimensionality} shows that any
minor-bidimensional problem belongs to the class \semips.
%Both of these algorithms provides a different tradeoff of space vs number of passes tradeoff for $k$-VC which uses $O(k^2)$ space and one pass.

Section~\ref{sec:lb} deals with lower bounds for parameterized streaming algorithms. First, in Section~\ref{sec:tight} we show that some parameterized problems are tight for the classes \semips and \bps. In particular, we show that $k$-Treewidth, $k$-Path and $k$-Feedback-Vertex-Set are tight for the class \semips, i.e., they belong to \semips but do not belong to the sub-class \subps. Our \semips algorithms are based on problem-specific structural insights. Via reductions from the \textsc{Perm} problem~\cite{woodruff-tight-bounds-insertion}, we rule out algorithms which use $\Ot(f(k)\cdot n^{1-\epsilon})$ bits (for any function $f$ and any $\epsilon\in (0,1)$) for these problems by proving $\Omega(n\log n)$ bits lower bounds for constant values of $k$.
%Interestingly, all of the graphs constructed in these lower bounds are planar graphs. Since planar graphs have $O(n)$ edges, it follows that the naive \semips algorithm of simply storing the whole graph is essentially optimal for these problems on planar graphs.
Then we show that some parameterized problems such as $k$-Girth and $k$-Dominating-Set are tight for the class \bps, i.e, they belong to \bps but do not belong to the sub-class \supps. Every graph problem belongs to \bps since we can store the entire adjacency matrix of the graph using $O(n^2)$ bits. Via reductions from the \textsc{Index} problem~\cite{nisan-cc-book}, we rule out algorithms which use $\Ot(f(k)\cdot n^{1+\epsilon})$ bits (for any function $f$ and any $\epsilon\in (0,1)$) for these problems by proving $\Omega(n^2)$ bits lower bounds for constant values of $k$.

Section~\ref{sec:dom-set-lb} shows a lower bound of $\Omega(n)$ bits for any algorithm that approximates (within a factor $\frac{\beta}{32}$) the size of min dominating set on graphs of arboricity $(\beta+2)$, i.e., this problem has no $\Ot(f(\beta)\cdot n^{1-\epsilon})$ bits algorithm (since $\beta$ is a constant), and hence does not belong to the class \subps when parameterized by $\beta$. In Section~\ref{sec:lb-from-parameterized} we obtain unconditional lower bounds on the space complexity of $1$-pass parameterized streaming algorithms for a large class of graph problems inspired by some of the recent frameworks to show conditional lower bounds for kernels~\cite{DBLP:journals/jcss/BodlaenderDFH09,dell-and,dell,DBLP:conf/focs/Drucker12,fortnow-santhanam}. Finally, in Section~\ref{app:lb-sat} we show that any parameterized streaming algorithm for the $d$-SAT problem (for any $d\geq 2$) must (essentially) follow the naive algorithm of storing all the clauses.
%This can be viewed as a ``streaming analogue" of Exponential Time Hypothesis (ETH).
%\todo{List the final set of results in above language}

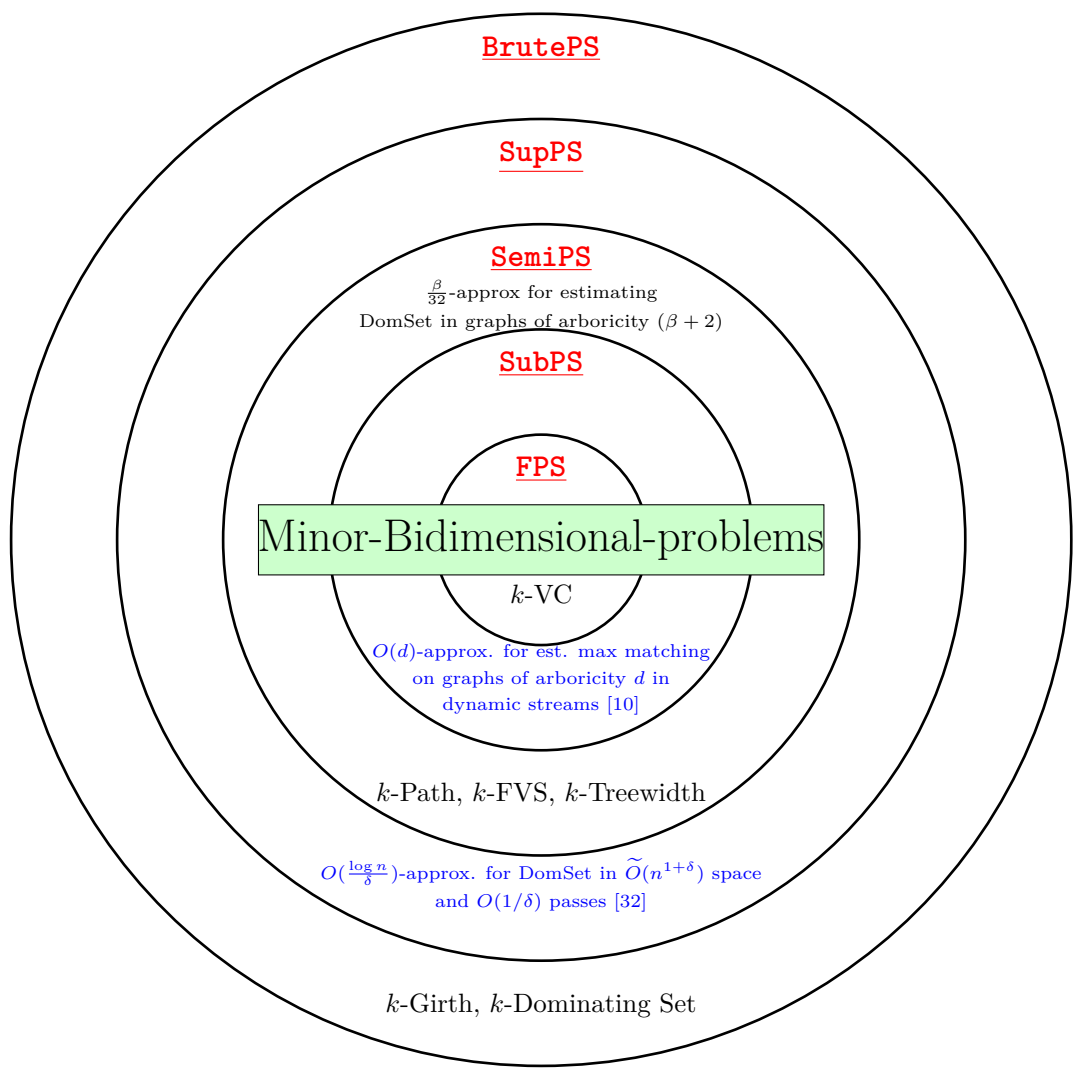
\begin{figure}[!ht]
  \centering

\begin{tikzpicture}[scale=0.93]
%[transform canvas={scale=1}]

\coordinate (originNode) at (90:0cm);
\coordinate (fps) at (90:1.5cm);
\coordinate (subps) at (90:3cm);
\coordinate (semips) at (90:4.5cm);
\coordinate (supps) at (90:6cm);
\coordinate (bps) at (90:7.5cm);

\draw[black,line width=1pt] (0,0) circle (1.5cm);
\draw[black,line width=1pt] (0,0) circle (3.0cm);
\draw[black,line width=1pt] (0,0) circle (4.5cm);
\draw[black,line width=1pt] (0,0) circle (6.0cm);
\draw[black,line width=1pt] (0,0) circle (7.5cm);

\node [below=0.5em] at (fps) {\textcolor[rgb]{1.00,0.00,0.00}{\Large \textbf{\underline{\fps}}}};
\node [below=0.5em] at (subps) {\textcolor[rgb]{1.00,0.00,0.00}{\Large \textbf{\underline{\subps}}}};
\node [below=0.5em] at (semips) {\textcolor[rgb]{1.00,0.00,0.00}{\Large \textbf{\underline{\semips}}}};
\node [below=0.5em] at (supps) {\textcolor[rgb]{1.00,0.00,0.00}{\Large \textbf{\underline{\supps}}}};
\node [below=0.5em] at (bps) {\textcolor[rgb]{1.00,0.00,0.00}{\Large \textbf{\underline{\bps}}}};

\coordinate (bpscaption) at (270:7.5cm);
\node [above=1.5em] at (bpscaption) {$k$-Girth, $k$-Dominating Set};

\coordinate (semipscaption) at (270:4.5cm);
\node [above=1.5em] at (semipscaption) {$k$-Path, $k$-FVS,$\newline$ $k$-Treewidth};

\coordinate (fpscaption) at (270:1.5cm);
\node [above=1.25em] at (fpscaption) {$k$-VC};

\coordinate (subpscaption) at (270:3cm);
\node [above=3em] at (subpscaption) {\textcolor[rgb]{0.00,0.00,1.00}{{\scriptsize $O(d)$-approx. for est. max matching}}};
\node [above=2em] at (subpscaption) {\textcolor[rgb]{0.00,0.00,1.00}{{\scriptsize on graphs of arboricity $d$ in}}};
\node [above=1em] at (subpscaption) {\textcolor[rgb]{0.00,0.00,1.00}{{\scriptsize dynamic streams~\cite{rajesh-soda-16}}}};

\coordinate (suppscaption) at (270:6cm);
\node [above=2.5em] at (suppscaption) {\textcolor[rgb]{0.00,0.00,1.00}{{\scriptsize $O(\frac{\log n}{\delta})$-approx. for DomSet in $\Ot(n^{1+\delta})$ space}}};
\node [above=1.5em] at (suppscaption) {\textcolor[rgb]{0.00,0.00,1.00}{{\scriptsize and $O(1/\delta)$ passes~\cite{sepideh}}}};

\draw [fill=green!20] (-4,-0.5) rectangle (4,0.5);
\node at (originNode) {{\LARGE Minor-Bidimensional-problems}};

\node [below=1.75em] at (semips) {{\scriptsize $\frac{\beta}{32}$-approx for estimating}};
\node [below=3em] at (semips) {{\scriptsize DomSet in graphs of arboricity $(\beta+2)$}};

\end{tikzpicture}
\caption{Pictorial representation of classification of some graph problems into complexity classes: our results are in black and previous work is referenced in blue. All results are for 1-pass deterministic algorithms on insertion-only streams unless otherwise specified. It was already known that $k$-VC $\in$ \fps~\cite{rajesh-soda-15,rajesh-soda-16} using only 1-pass, but here we design an algorithm with optimal space storage at the expense of multiple passes.\label{landscape}}
\end{figure}

\begin{figure}[!ht]
\begin{center}
    \begin{tabular}{ | c | c | c | c | c |}
    \hline
     \textbf{Problem} & \Longstack{\textbf{Number}\\ \textbf{of}\\ \textbf{Passes}} & \Longstack{\textbf{Type of}\\ \textbf{Stream}} & \Longstack{\textbf{Space}\\ \textbf{Upper Bound}} & \Longstack{\textbf{Space}\\ \textbf{Lower Bound}}\\ \hline
     %$k$-VC & 1 & Insertion only & $\Ot(k^2)$ & $\Omega(k^2)$ \\ \hline
     %$k$-VC & 1 & Insertion-Deletion & $\Ot(k^2)$ & $\Omega(k^2)$ \\ \hline
     \Longstack{\small{$g(r)$-minor-bidimensional}\\ \small{problems} [Sec.~\ref{subsec:bidimensionality}]}  & 1 & Ins-Del. & \Longstack{\scriptsize{$\Ot((g^{-1}(k+1))^{10} n)$} \\ words}   & --- \\\hline
     $k$-VC [Sec.~\ref{subsec:ic}]       & \Longstack{$2^{2k}\cdot k$} & Ins-only & $O(k)$ words & $\Omega(k)$ words\\ \hline
     $k$-VC [Sec.~\ref{subsec:branching}]      & $2^{k}$ & Ins-only & $O(k)$ words  & $\Omega(k)$ words\\ \hline
     \Longstack{$k$-FVS, $k$-Path\\ $k$-Treewidth [Sec.~\ref{sec:tight}]} & 1 & Ins-only & $O(k\cdot n)$ words & \Longstack{No $f(k)\cdot n^{1-\epsilon}\log^{O(1)} n$\\ bits algorithm} \\\hline
     \Longstack{$k$-FVS, $k$-Path\\ $k$-Treewidth [Sec.~\ref{sec:tight}]} & 1 & Ins-Del. & $\Ot(k\cdot n)$ words & \Longstack{No $f(k)\cdot n^{1-\epsilon}\log^{O(1)} n$\\ bits algorithm} \\\hline
     \Longstack{$k$-Girth, $k$-DomSet,\\ [Sec.~\ref{sec:tight}]} & 1 & Ins-Del. & $O(n^2)$ bits & \Longstack{No $f(k)\cdot n^{2-\epsilon}\log^{O(1)} n$\\ bits algorithm} \\\hline
     \Longstack{$\frac{\beta}{32}$-approximation for size of \\ min DomSet on graphs of \\ arboricity $\beta$ [Sec.~\ref{sec:dom-set-lb}]} & 1 & Ins-only & $\Ot(n\beta)$ bits & \Longstack{No $f(\beta)\cdot n^{1-\epsilon}$\\ bits algorithm} \\\hline

     \Longstack{AND-compatible problems\\ and OR-compatible\\ problems [Sec.~\ref{sec:lb-from-parameterized}]} & 1 & Ins-only & $O(n^2)$ bits &  \Longstack{No $\Ot(f(k)\cdot n^{1-\epsilon})$\\ bits algorithm} \\\hline
    %\Longstack{FPT problems with\\ no poly kernels [Sec.~\ref{sec:lb-from-parameterized}]} & 1 &  Ins-only & $O(n^2)$ bits & \Longstack{(*) No $\Ot(k^{O(1)})$\\ bits algorithm}  \\\hline

     \Longstack{$d$-SAT with\\ $N$ variables  [Sec.~\ref{app:lb-sat}]} & 1 & \Longstack{Clause \\ Arrival} & $\Ot(d\cdot N^d)$ bits & $\Omega((N/d)^d)$ bits \\\hline
     %SAT & 1 & Clause Arrival & $\Ot(n\cdot 2^n)$ & $\Omega(2^n)$ \\\hline

    %Bidimensional\\ Parameters on\\ $H$-minor-free graphs & 1 & Insertion Deletion & $\Ot(k\cdot n)$ & ? \\\hline

    \end{tabular}
\end{center}
\caption{Table summarizing our results (in the order in which they appear in the paper). All our algorithms are deterministic. All the lower bounds are unconditional, and hold even for randomized algorithms in insertion-only streams.
%The two lower bounds marked with (*) are deterministic and conditional.
}
\label{table}
\end{figure}

Figure~\ref{landscape} provides a pictorial representation of the complexity classes, and the known classification of several graph problems (from this paper and some previous work) into these classes.
Figure~\ref{table} summarizes our results, and clarifies the
stream arrival model(s) under which they hold. Figure~\ref{vc-case-study} summarizes known results for the $k$-VC problem in the different settings outlined in Remark~\ref{rem:7-tuple}.
%Figure~\ref{vc-case-study} summarizes results for the $k$-VC problem in the different settings outlined in Remark~\ref{rem:7-tuple}.

\begin{figure}[!ht]
\begin{center}
    \begin{tabular}{ | c | c | c | c | c | c |}
    \hline
     \textbf{Problem} & \Longstack{\textbf{\# of}\\ \textbf{Passes}} & \Longstack{\textbf{Type of}\\ \textbf{Stream}} & \Longstack{\textbf{Type of}\\ \textbf{Algorithm}} & \Longstack{\textbf{Approx.}\\ \textbf{Ratio}} & \Longstack{\textbf{Space}\\ \textbf{ Bound}}\\ \hline
     %$k$-VC & 1 & Insertion only & $\Ot(k^2)$ & $\Omega(k^2)$ \\ \hline
     %$k$-VC & 1 & Insertion-Deletion & $\Ot(k^2)$ & $\Omega(k^2)$ \\ \hline
     $k$-VC  & 1 & Ins-only & Det. & 1  & $O(k^2 \log n)$ bits~\cite{rajesh-soda-15} \\\hline

     $k$-VC  & 1 & Ins-only & Rand. & 1   & $\Omega(k^2)$ bits~\cite{rajesh-soda-15} \\\hline

     $k$-VC  & 1 & Ins-Del. & Rand. & 1  & $O(k^2 \log^{O(1)} n)$ bits~\cite{rajesh-soda-16} \\\hline

     $k$-VC  & $2^k$ & Ins-only & Det. & 1  & $O(k \log n)$ bits [Algorithm~\ref{alg:streaming-branching}] \\\hline

     $k$-VC  & $k\cdot 2^k$ & Ins-only. & Det. & 1   & $O(k \log n)$ bits [Algorithm~\ref{alg:streaming-ic}] \\\hline

     Estim. $k$-VC  & $\Omega(k/\log n)$ & Ins-only. & Rand. & 1   & $O(k\log n)$ bits~\cite[Theorem 16]{abboud} \\\hline

     \Longstack {Estim. $k$-VC \\ on Trees}  & $1$ & Ins-only. & \Longstack{Det. \\ Rand.} & $(3/2-\epsilon)$   & \Longstack{$\Omega(n)$ bits~\cite[Theorem 6.1]{hossein-soda-15}\\ $\Omega(\sqrt{n})$ bits~\cite[Theorem 6.1]{hossein-soda-15}}  \\\hline

    \end{tabular}
\end{center}
\caption{Table summarizing some of the results for the $k$-VC problem in the different settings outlined in Remark~\ref{rem:7-tuple}.
%The two lower bounds marked with (*) are deterministic and conditional.
}
\label{vc-case-study}
\end{figure}

%\subsection{Related Prior Work}
%
%
%
%Our 3 papers + Sublinear estimation of matching on bounded arboricity graphs
%
%\todo{Any other papers?}

%\subsection{Preliminaries}
%
%\todo{Currently does not seem like we need this subsection. All definitions and notation are introduced later when needed}

%\vspace{-30mm}

\subsection{Prior work on Parametrized Streaming Algorithms}

Prior work began by considering how to implement kernels in the
streaming model.
Formally, a streaming kernel~\cite{fafianie-kratsch} for a
parameterized problem $(I,k)$ is a streaming algorithm that receives
the input $I$ as a stream of elements, stores $f(k)\cdot \log^{O(1)}
|I|$ bits and returns an equivalent
instance\footnote{~\cite{fafianie-kratsch} required $f(k)=k^{O(1)}$,
  but we choose to relax this requirement}.
 This is especially important from the practical point of view since several real-world situations can be modeled by the streaming setting, and streaming kernels would help to efficiently preprocess these instances. Fafianie and Kratsch~\cite{fafianie-kratsch} showed that the kernels for some problems like Hitting Set and Set Matching can be implemented in the streaming setting, but other problems such as Edge Dominating Set, Feedback Vertex Set, etc. do not admit (1-pass) streaming kernels.

Chitnis et al.~\cite{rajesh-soda-15} studied how to circumvent the
worst case bound of $\Omega(n^2)$ bits for Vertex Cover by designing a
streaming algorithm for the parameterized $k$-Vertex-Cover
($k$-VC)\footnote{That is, determine   whether there is a vertex cover
  of size at most $k$?}.
%which uses $f(k)\cdot g(n)$ bits for some subquadratic function $g$?
%This would give an improvement over the current best algorithm for
%Vertex Cover which uses $\Ot(n^2)$ bits and simply stores the entire
%graph.
%answered this question in the
%affirmative by showing
They showed that the $k$-VC problem can be solved in
insertion-only streams using storage of $O(k^2)$ space.
They also showed an almost matching lower bound of $\Omega(k^2)$ bits
for any streaming algorithm for $k$-VC.
A sequence of papers showed how to solve the $k$-VC problem in more
general streaming models: Chitnis et
al.~\cite{rajesh-soda-15,rajesh-spaa-15} gave an $\Ot(k^2)$ space
algorithm under a particular promise,
which was subsequently removed in~\cite{rajesh-soda-16}.

%\emph{promised streams}\footnote{These are streams with both edge insertions and deletions. However, it is promised that the vertex cover never exceeds $k$}. Finally, the requirement of the \emph{promise} was removed and an $\Ot(k^2)$ space algorithm for $k$-VC was designed in for general insertion-deletion streams.

Recently, there have been several papers considering the problem of
estimating the size of a maximum matching using $o(n)$ space in graphs
of bounded arboricity.
%\footnote{A graph has arboricity $\alpha$ if its edges can be
%partitioned into $\alpha$ forests}.
If the space is required to be sublinear in $n$, then versions of the
problem that involve estimating the size of a maximum matching (rather
than demonstrating such a matching) become the focus.
%If we only allow sublinear space, then one necessarily has to consider
%the problem of estimating the size of (instead of reporting) a maximum
%matching since paths have arboricity $1$ but have a matching of size
%$\Omega(n)$.
Since the work of Esfandiari et al.~\cite{hossein-soda-15}, there have
been several sublinear space
algorithms~\cite{andrew-approx,andrew-sosa,graham-sparse,rajesh-soda-16}
which obtain $O(\alpha)$-approximate estimations of the size of
maximum matching in graphs of arboricity $\alpha$.
The current best bounds~\cite{Bury2018,graham-sparse} for insertion-only streams is
$O(\log^{O(1)} n)$ space and for insertion-deletion streams is
$\Ot(\alpha\cdot n^{4/5})$.
All of these results can be viewed as parameterized streaming algorithms (\fps or \subps) for approximately estimating the size of maximum matching in graphs parameterized by the arboricity.

\section{\mbox{Parameterized Streaming Algorithms Inspired by FPT techniques}}
\label{sec:algo-fpt-to-streaming}

In this section we design parameterized streaming algorithms using three techniques from the world of parameterized algorithms, viz. branching, iterative compression and bidimensionality.

\subsection{Multipass \fps algorithm for $k$-VC using Branching}
\label{subsec:branching}

The streaming algorithm (Algorithm~\ref{alg:streaming-ic}) from Section~\ref{subsec:ic} already uses optimal storage of $O(k\log n)$ bits but requires $O(2^{k}\cdot (n-k))$ passes. In this section, we show how to reduce the number of passes to $2^k$ (while still maintaining the same storage) using the technique of bounded-depth search trees (also known as branching).
The method of bounded-depth search trees gives a folklore FPT algorithm for $k$-VC which runs in $2^{O(k)}\cdot n^{O(1)}$ time. The idea is simple: any vertex cover must contain at least one end-point of each edge. We now build a search tree as follows: choose an arbitrary edge, say $e=u-v$ in the graph. Start with the graph $G$ at the root node of the search tree. Branch into two options, viz. choosing either $u$ or $v$ into the vertex cover\footnote{Note that if we choose $u$ in the first branch then that does not imply that we cannot or will not choose $v$ later on in the search tree}. The resulting graphs at the two children of the root node are $G-u$ and $G-v$. Continue the branching process. Note that at each step, we branch into two options and we only need to build the search tree to height $k$ for the $k$-VC problem. Hence, the binary search tree has $2^{O(k)}$ leaf nodes. If the resulting graph at any leaf node is empty (i.e., has no edges) then $G$ has a vertex cover of size $\leq k$ which can be obtained by following the path from the root node to the leaf node in the search tree. Conversely, if the resulting graphs at none of the leaf nodes of the search tree are empty then $G$ does not have a vertex cover of size $\leq k$: this is because at each step we branched on all the (two) possibilities at each node of the search tree.

%\subsubsection{Simulating branching-based FPT algorithm using multiple passes}
%\label{subsec:vc-branching}
\textbf{Simulating branching-based FPT algorithm using multiple passes}: We now simulate the branching-based FPT algorithm described in the previous section using $2^{k}$ passes and $O(k\log n)$ bits of storage in the streaming model.

\begin{definition}
Let $V(G)=\{v_1, v_2, \ldots, v_n\}$. Fix some ordering $\phi$ on $V(G)$ as follows: $v_1 < v_2 < v_3 < \ldots< v_n$.
%Note that we can maintain using $n$ bits the following:
%\begin{itemize}
%  \item Which subset we are currently looking at. This is denoted by $\curr(\dict_n)$
%  \item Which subset comes next in $\dict_n$ after the current subset $X$. This is denoted by $\nextt(X)$
%\end{itemize}
%\end{dfn}
%
%\begin{dfn}
Let $\dict_k$ be the dictionary ordering on the $2^{k}$ binary strings of $\{0,1\}^{k}$. Given a string $X\subseteq \{0,1\}^{k}$, let $\dict_{k}(\nextt(X))$ denote the string that comes immediately after $X$ in the ordering $\dict_k$. We set $\dict_{k}(\nextt(1^k))=\spadesuit$
%We define $\dict_{k}(\nextt(U))=U$
%Note that we can maintain using $n$ bits the following:
%\begin{itemize}
%  \item Which subset we are currently looking at. This is denoted by $\curr(\dict_n)$
%  \item Which subset comes next in $\dict_n$ after the current subset $X$. This is denoted by $\nextt(X)$
%\end{itemize}
\end{definition}

We formally describe our multipass algorithm in Algorithm~\ref{alg:streaming-branching}. This algorithm crucially uses the fact that in each pass we see the edges of the stream in the \emph{same} order.

\begin{algorithm}[h]
\caption{\small{$2^{k}$-pass Streaming Algorithm for $k$-VC using $O(k\log n)$ bits via Branching} \label{alg:streaming-branching}}
\textbf{Input:} \mbox{An undirected graph $G=(V,E)$ and an integer $k$.}\\
\textbf{Output:} \mbox{A vertex cover $S$ of $G$ of size $\leq k$ (if one exists), and NO otherwise}\\
\textbf{Storage:} $i$, $j$, $S$, $X$
%~\\

\begin{algorithmic}[1]

%    \If {$\text{foo} \leq \text{bar}$} $\text{doh} = 0$
%    \Else {} $\text{duh} = 1$
%    \EndIf
%    \IfThenElse {$\text{foo} \leq \text{bar}$}% If ...
%      {$\text{doh} = 0$}% ...then...
%      {$\text{duh} = 1$}% ...else...
%
%\STATE Let $S=\emptyset$
\State Let $X=0^k$, and suppose the edges of the graph are seen in the order $e_1, e_2, \ldots, e_m$

\While{$X\neq \spadesuit$}

        {$S=\emptyset, i=1, j=1$}
        \While{$i\neq k+1$}
             \State Let $e_j = u-v$ such that $u<v$ under the ordering $\phi$
             \If{Both $u\notin S$ and $v\notin S$}

                \If{$X[i] =0$} $S\leftarrow S\cup \{u\}$
                    \Else {}  $S\leftarrow S\cup \{v\}$
                \EndIf
                \State $i\leftarrow i+1$
             \EndIf
             \State $j\leftarrow j+1$
        \EndWhile
        \If{$j=m+1$} $\text{Return}\ S\ \text{and abort}$
        \Else\ $X\leftarrow \dict_{k}(\nextt(X))$
        \EndIf
\EndWhile
        \If{$X=\spadesuit$} $\text{Return NO}$
        \EndIf

\end{algorithmic}
\end{algorithm}

%\begin{algorithm}
%  \caption{Single line IF THEN ELSE}
%  \begin{algorithmic}[1]
%
%    \If {$\text{foo} \leq \text{bar}$} $\text{doh} = 0$
%    \Else {} $\text{duh} = 1$
%    \EndIf
%    \IfThenElse {$\text{foo} \leq \text{bar}$}% If ...
%      {$\text{doh} = 0$}% ...then...
%      {$\text{duh} = 1$}% ...else...
%  \end{algorithmic}
%\end{algorithm}

\begin{theorem}%$[\star]$
\label{thm:alg-vc-branching}
Algorithm~\ref{alg:streaming-branching} correctly solves the $k$-VC problem using $2^{k}$ passes and $O(k\log n)$ bits of storage.
\end{theorem}
\begin{proof}
First we argue the correctness of Algorithm~\ref{alg:streaming-branching}. Suppose that there is a string $X\in \{0,1\}^{k}$ such that $j=m+1$ and we return the set $S$. Note that initially we have $S=\emptyset$, and the counter $i$ increases each time we add a vertex to $S$. Hence, size of $S$ never exceeds $k$. Moreover, if an edge was not covered already (i.e., neither endpoint was in $S$) then we add at least one of those end-points in $S$ (depending on whether $X[i]$ is $0$ or $1$) and increase $i$ by one. Hence, if $j=m+1$ then this means that we have seen (and covered) all the edges and the current set $S$ is indeed a vertex cover of size $\leq k$. Now suppose that the algorithm returns NO. We claim that indeed $G$ cannot have a vertex cover of size $k$. Suppose to the contrary that $G$ has a vertex cover $S^*$ of size $\leq k$, but Algorithm~\ref{alg:streaming-branching} returned NO. We construct a string $X^{*}\in \{0,1\}^{k}$ for which Algorithm~\ref{alg:streaming-branching} would return the set $S^*$: we start with $i$=1 and the edge $e_1$. Since $S^*$ is a vertex cover of $G$ it must cover the edge $e_1$. Set $X^{*}[1]$ to be 0 or 1 depending on which of the two endpoints of $e_1$ is in $S^*$ (if both endpoints are in $S^*$, then it does not matter what we set $X^{*}[1]$ to be). Continuing this way suppose we have filled the entries till $X^{*}[i]$ and the current edge under consideration is $e_j$. If $e_j$ is not covered then $i\neq k$ since $S^*$ is a vertex cover of $G$ of size $\leq k$. In this case, we set $X^{*}[i+1]$ to be 0 or 1 depending on which of the two endpoints of $e_1$ is in $S^*$.

We now analyze the storage and number of passes required. The number
of passes is at most $2^{k}$ since we have one pass for each string from $\{0,1\}^{k}$. During each pass, we store four quantities:
\begin{itemize}
  \item The string $X\in \{0,1\}^{k}$ under consideration in this pass. This needs $k$ bits.
  \item The index $i$ of current bit of the $k$-bit binary string $X$ under consideration in this pass. This needs $\log k$ bits.
  \item The index $j$ of the current edge under consideration in this pass. This needs $\log n$ bits.
  \item The set $S$. Since size of $S$ never exceeds $k$ throughout the algorithm, this can be done using $k\log n$ bits.
  %\item The current subset $Y\subseteq S$ under consideration for being the intersection of $S$ and new potential VC of size $\leq k$. Since $|S|\leq k+1$ and we store $S$ explicitly, it follows that we can store $Y$ and find $\nextt(Y)$ using $k\log n$ bits
\end{itemize}
\end{proof}

Note that the total storage of Algorithm~\ref{alg:streaming-branching} is $O(k\log n)$ bits which is essentially optimal since the algorithm also outputs a vertex cover of size at most $k$ (if one exists).
The next natural question is whether one need exponential (in $k$) number of passes when we want to solve the $k$-VC problem using only $O(k\log n)$ bits. A lower bound of $(k/\log n)$ passes follows for such algorithms from the following result of Abboud et al.
\begin{theorem}\emph{(rewording of ~\cite[Thm 16]{abboud})}
Any algorithm for the $k$-VC problem which uses $S$ bits of space and $R$ passes must satisfy $RS \geq n^2$
\label{thm:abboud-lb}
\end{theorem}

\subsection{Multipass \fps algorithm for $k$-VC using Iterative Compression}
\label{subsec:ic}

The technique of \emph{iterative compression} was introduced by Reed et al.~\cite{reed-smith-vetta-ic} to design the first FPT algorithm for the $k$-OCT problem\footnote{Is there a set of size at most $k$ whose deletion makes the graph odd cycle free, i.e. bipartite}. Since then, iterative compression has been an important tool in the design of faster  parameterized algorithms~\cite{dfvs-jacm,rajesh-talg,chen-fvs-ic} and kernels~\cite{ic-kernel}. In this section, using the technique of iterative compression, we design an algorithm (Algorithm~\ref{alg:streaming-ic}) for $k$-VC which uses $O(k\log n)$ bits but requires $O(k\cdot 2^{2k})$ passes. Although this algorithm is strictly worse (same storage, but higher number of passes) compared to Algorithm~\ref{alg:streaming-branching}, we include it here to illustrate that the technique of iterative compression can be used in the streaming setting.

\subsubsection{FPT algorithm for $k$-VC using iterative compression}
\label{app:ic-fpt}

We first define a variant problem where we are given some additional information in the input in the form of a vertex cover of size of size $k+1$ (just more than the budget).
%\begin{center}
%\noindent\framebox{\begin{minipage}{6.00in}
%\textbf{Vertex Cover (VC)}\\
%\emph{Input}: A graph $G$, and a positive integer $k$\\
%\emph{Parameter}: $k$\\
%\emph{Question}: Does there exist a set $X\subseteq V(G)$ with $|X|\leq k$ such that $G\setminus X$ has no edges?
%\end{minipage}}
%\end{center}

\begin{center}
\noindent\framebox{\begin{minipage}{5in}
\textbf{\textsc{Compression-VC}}\\
\emph{Input}: A graph $G$, a positive integer $k$ and a vertex cover $T$ of size $k+1$\\
\emph{Parameter}: $k$\\
\emph{Question}: Does there exist a set $X\subseteq V(G)$ with $|X|\leq k$ such that $G\setminus X$ has no edges?
\end{minipage}}
\end{center}

\begin{lemma}%$[\star]$\footnote{The proofs of the results labeled with $\star$ have been deferred to the full version of the paper.}
[\textbf{power of iterative compression}]
 $k$-VC can be solved by $k$ calls to an algorithm for the
\textsc{Compression-VC} problem.
\label{lem:ic}
\end{lemma}
\begin{proof}
Let $e_1, e_2, \ldots, e_t$ be the edges of a maximal matching $M$ in $G$, and let $V_M$ be the set of vertices which are matched in $M$ . If $t>k$ then there is no vertex cover of size $k$ since any vertex cover needs to pick at least one vertex from every edge of the maximal matching. Hence, we have $t\leq k$. By maximality of $M$, it follows that the set $V_M$ forms a vertex cover of size $2t\leq 2k$. For each $2k\geq r\geq k+1$ we now run the \textsc{Compression-VC} problem to see whether there exists a vertex cover of size $r-1$. If the answer is YES, then we continue with the compression. On the other hand, if the the \textsc{Compression-VC} problem answers NO for some $2k\geq r\geq k+1$ then clearly there is no vertex cover of $G$ which has size $\leq k$.
%Let $V(G)=\{v_1,\ldots,v_n\}$ and for $i\in [n]$ let $V_i = \{v_1, \ldots v_i\}$. We construct a sequence of subsets $X_i
%\subseteq V_i$, such that $X_i$ is a solution for $G[V_i]$. Clearly, $X_1=\emptyset$ is a solution for $G[V_1]$. Observe that
%if $X_i$ is a solution for $G[V_i]$, then $X_i \cup \{v_{i+1}\}$ is a solution for $G[V_{i+1}]$. Therefore, for each $i\in
%[n-1]$, we set $T = X_i \cup \{v_{i+1}\}$ and use, as a blackbox, an algorithm for \textsc{Compression-VC}, to construct a
%set $X_{i+1}$ that is a solution of size at most $k$ for $G[V_{i+1}]$. Note that if there is no solution for $G[V_i]$ for
%some $i\in [n]$, then there is no solution for the whole graph $G$ and moreover, since $V_n = V(G)$, if all the calls to the
%reduction problem are successful, then $X_n$ is a solution of size at most $k$ for the graph $G$.
\end{proof}

Now we solve the \textsc{Compression-VC} problem via the following problem whose only difference is that the vertex cover in the output must be disjoint from the one in the input:

\begin{center}
\noindent\framebox{\begin{minipage}{5.00in}
\textbf{\textsc{Disjoint-VC}}\\
\emph{Input}: A graph $G$, a positive integer $k$ and a vertex cover $T$ of size $k+1$\\
\emph{Parameter}: $k$\\
\emph{Question}: Does there exist a set $X\subseteq V(G)$ with $|X|\leq k$ such that $X\cap T=\emptyset$ and $G\setminus X$ has no edges?
\end{minipage}}
\end{center}

\begin{lemma}[\textbf{adding disjointness}] \textsc{Compression-VC} can be solved by $O(2^{|T|})$ calls to an algorithm for the
\textsc{Disjoint-VC} problem. \label{lem:disjoint}
\end{lemma}
\begin{proof}
  Given an instance $I=(G,T,k)$ of \textsc{Compression-VC} we
  guess the intersection $Y$ of the given vertex $T$ of size $k+1$ and the desired vertex cover
  $X$ of size $k$ in the output. We have at most $2^{|T|}-1$ choices for $Y$ since we can have all possible subsets of $T$ except $T$ itself. Then
  for each guess for $Y$, we solve the \textsc{Disjoint-VC} problem for the instance $I_Y=(G\setminus
  Y,T\setminus Y,k-|Y|)$. It is easy to see that if $X$ is a
  solution for instance $I$ of \textsc{Compression-VC}, then
  $X\setminus Y$ is a solution of instance $I_Y$ of \textsc{Disjoint-VC} for $Y=T\cap X$. Conversely, if $Z$ is a
  solution to some instance $I_Y=(G\setminus Y,T\setminus Y,k-|Y|)$ of \textsc{Disjoint-VC}, then $Z\cup Y$ is a solution for the instance $I=(G,T,k)$ of \textsc{Compression-VC}.
\end{proof}

Using a maximal matching, we either start with a vertex cover of size $\leq 2k$ or we can answer that $G$ has no vertex cover of size $\leq k$. Hence, any algorithm for \textsc{Disjoint-VC} gives an algorithm for
the $k$-VC problem, with an additional blowup of $O(2^{2k}\cdot k)$.
Since our objective is to show that the $k-$VC problem is FPT, then it
is enough to give an FPT algorithm for the \textsc{Disjoint-VC}
problem (which has additional structure that we can exploit!).
In fact we show that the \textsc{Disjoint-VC} problem can be solved in polynomial time.
%\section{Properties needed to apply Iterative Compression}
%\begin{itemize}
%\item So far we only know how to apply Iterative Compression for minimization problems
%\item Our problem should be such that if $H$ is a subgraph of $G$ such that $H$ does not have a solution of size $\leq k$ then we should be able to conclude that $G$ also does not have a solution of size $k$. Most graph problems are ``hereditary", and hence satisfy this property strongly: if $H$ is a subgraph of $G$ and $X$ is a solution for $G$ then $X\cap H$ is a solution for $H$.
%
%\item We should be able to generate an ``initial" solution of size $k+1$ to start the compression process. In the case of VC, it was trivial since we just took the subgraph induced by some $k+1$ vertices.
%
%\item It is not enough for the Compression version to be a decision problem: it should actually output a compressed solution since we need it to build a new solution of size $k+1$ which we again use as input to the Compression problem.
%
%\end{itemize}
\begin{lemma}
The \textsc{Disjoint-VC} problem can be solved in polynomial time.
\end{lemma}
\begin{proof}
Let $(G,T,k)$ be an instance of \textsc{Disjoint-VC}. Note that
$G\setminus T$ has no edges since $T$ is a vertex cover.
Meanwhile, if $G[T]$ has even a single edge, then answer is NO since we cannot pick any vertices from $T$ in the vertex cover. So the only edges are between $T$ and $G\setminus T$. Since we cannot pick any vertex from $T$ in vertex cover, we are forced to pick all vertices in $G\setminus T$ which have neighbors in $T$. Formally, we have to pick the set $X=\{x\notin T\ :\ \exists y\in T\ \text{such that}\ x-y\in E(G) \}$. Note that picking $X$ is both necessary and sufficient. So it simply remains to compare $|X|$ with $k$ and answer accordingly.
\end{proof}

Consequently, we obtain a $O(2^{2k}\cdot k\cdot n^{O(1)})$ time algorithm for $k$-VC
by composing these two reductions.

\subsubsection{Simulating the FPT algorithm in streaming using multiple passes}
\label{app:ic-streaming}

In this section, we show how to simulate the FPT algorithm of the previous section in the
multi-pass streaming model.
First, let us fix some order on all subsets of $[n]$.

\begin{definition}
Let $U=\{u_1, u_2, \ldots, u_n\}$ and $k\leq n$. Let $\mathcal{U}_k$ denote the set of all $\sum_{i=0}^{k} \binom{|U|}{i}$ subsets of $U$ which have at most $k$ elements, and $\dict_{U_k}$ be the dictionary ordering on $\mathcal{U}_k$. Given a subset $X\in \mathcal{U}_k$, let $\dict_{U_k}(\nextt(X))$ denote the subset that comes immediately after $X$ in the ordering $\dict_U$. We denote the last subset in the dictionary order of
$\mathcal{U}_k$ by $\text{Last}(\mathcal{U}_k)$ and use the notation that $\dict_{U_k}(\text{Last}(\mathcal{U}_k))=\spadesuit$.
%Note that we can maintain using $n$ bits the following:
%\begin{itemize}
%  \item Which subset we are currently looking at. This is denoted by $\curr(\dict_n)$
%  \item Which subset comes next in $\dict_n$ after the current subset $X$. This is denoted by $\nextt(X)$
%\end{itemize}
\end{definition}

\begin{algorithm}[t]
\caption{Multipass Streaming Algorithm for $k$-VC using Iterative Compression \label{alg:streaming-ic}}
\textbf{Input:} \mbox{An undirected graph $G=(V,E)$, integer $k$.}\\
\textbf{Output:} \mbox{A vertex cover $S$ of $G$ of size at most $k$ (if one exists), and NO otherwise}\\
\textbf{Storage:} $i$, $S$, $Y$, $V_M$
%\\
%~\\
\begin{algorithmic}[1]

\State Find a maximal matching $M$ (upto size $k$) in 1 pass which saturates the vertices $V_M$ \label{step:mm}

\State If $|M|$ exceeds $k$, then return NO and abort

\State Let $S=V_M$

\For{$i=|V_M|$ to $k+1$}
        %\State $S\leftarrow S\cup v_{i+1}$
        \State $Y=\emptyset$
        \While{$Y\in \mathcal{S}_k, Y\neq \spadesuit$}
             \If{$|\{x\in V\setminus S\ :\ \exists y\in S\setminus Y\ \text{s.t.}\ \{x,y\}\in E(G)\}|\leq k-|Y|$} \label{step:disjoint-VC}
                \State $S\leftarrow Y\cup \{x\in V\setminus
                S\ :\ \exists y\in S\setminus Y\ \text{s.t.}\ x-y\in
                E(G)\}$  \label{step:disj-VC}
                \Comment{Requires a pass through the data}
                \State Break % $i\leftarrow i+1$
                \Comment{Found a solution, and reduce value of $i$ by $1$}
             \Else
             \State $Y\leftarrow \dict_{\mathcal{S}_k}(\nextt(Y))$
             \Comment{Try the next subset}
             \EndIf
        \EndWhile
        \If{$Y=\spadesuit$}
            \State Return NO and abort
        \EndIf

\EndFor

        \If{$i=k$}
            \State Return $S$
        \EndIf

%\STATE

%\STATE

%\STATE Let $Z=Z_1\cup Z_2$.
%
%\STATE Let $G_3=\torso(G_1,V(G)\setminus Z)$.

\end{algorithmic}
\end{algorithm}

We give our multipass algorithm as Algorithm~\ref{alg:streaming-ic},
whose correctness follows from Section~\ref{app:ic-fpt}.
We now analyze the storage and number of passes required.

We first use one pass to store a maximal matching $M$ (upto $k$ edges). The remaining number
of passes used by the algorithm is at most $2^{2k}\cdot k=O(2^{2k}\cdot k)$ since we
have $(k)$ iterations over the index $i$, we have $2^{2k}$
choices for the set $Y\in \mathcal{S}_k$ (since $|S|\leq 2k$) and we need one pass for each execution of
Step~\ref{step:disjoint-VC} and Step~\ref{step:disj-VC}. Throughout the algorithm, we store three quantities:
\begin{itemize}
    \item We store the vertices $V_M$  saturated by a maximal matching $M$ (but only until the size of $M$ exceeds $k$ in which case we output NO). This needs at most $2k\log n$ bits
  \item The index $i$ of current iteration. This needs $\log n$ bits.
  \item The set $S$. Since size of $S$ never exceeds $2k$ throughout the algorithm, this can be done using $2k\log n$ bits.
  \item The current subset $Y\subseteq \mathcal{S}_k$ under consideration for being the intersection of $S$ and new potential VC of size $\leq k$. Since $|S|\leq 2k$ and we store $S$ explicitly, it follows that we can store $Y$ and find $\nextt(Y)$ using $O(k\log n)$ bits.
\end{itemize}

Hence, the total storage of the algorithm is $O(k\log n)$ bits which
is essentially optimal since the algorithm also outputs a vertex cover
of size at most $k$ (if one exists).

\subsection{Minor-Bidimensional problems belong to \semips}
\label{subsec:bidimensionality}

The  theory of  bidimensionality~\cite{mth-bidimensionality-jacm,mth-bidimensionality-survey} provides a general technique for designing (subexponential) \FPT for \NP-hard graph  problems on various graph classes. First, we introduce some graph theoretic concepts.

\begin{definition}[\textbf{treewidth}]
Let $G$ be a given undirected graph. Let $\T$ be a tree and
$B:V(\T)\rightarrow 2^{V(G)}$. The pair $(\T,B)$ is a \emph{tree
  decomposition} of an undirected graph $G$ if every vertex $x\in
V(\T)$ of the tree $\T$
has an assigned set of vertices $B_x \subseteq V(G)$ (called a bag) such that the following properties are satisfied:
\begin{itemize}
\item $\textbf{(P1)}$: $\bigcup_{x\in V(\T)} B_x = V (G)$.
\item $\textbf{(P2)}$: For each $\{u,v\}\in E(G)$, there exists an $x\in V(\T)$ such that $u, v\in B_x$.
\item $\textbf{(P3)}$: For each $v\in  V(G)$, the set of vertices of $\T$ whose bags contain $v$ induce a connected subtree of $\T$.
\end{itemize}
The \emph{width} of a tree decomposition $(\T,B)$ is $\max_{x\in V(\T)} |B_x|-1$. The treewidth of a graph $G$, usually denoted by $\tw(G)$, is the minimum width over all tree decompositions of $G$.
\end{definition}

Intuitively, the treewidth of a graph captures how tree-like it is. Trees (and more generally forests) have treewidth $1$.
%Theorem~\ref{thm:excluded-minor} gives a WIN-WIN approach as follows: either we have a small treewidth (and we can do dynamic programming), or otherwise the graph contains a large grid which gives allows for some structural insights.

\begin{definition}[\textbf{minor}]
  Let $H,G$ be two undirected graphs. We say that $H$ is a minor of $G$ if $H$ can be obtained from $G$ by a sequence of edge deletions, vertex deletions or edge contractions.
\end{definition}

%\begin{dfn}
%A parameter $\Pi$ is said to be contraction-bidimensional parameter if
%\begin{itemize}
%  \item Value of $\Pi$ on the $r\times r$ grid is $\Omega(r^2)$
%  \item $Pi$ is closed under taking minors, i.e., value of $\Pi$ does not increase under the operations of vertex deletions, edge deletions, edge contractions.
%\end{itemize}
%\end{dfn}

One of the foundational results of graph theory is the Excluded Grid Minor Theorem of Robertson and Seymour~\cite{rs86} which states that large treewidth forces large grid minors:
\begin{theorem}
\label{thm:excluded-minor}
\emph{~\cite{rs86}} There is a function $f: \mathbb{N}\rightarrow \mathbb{N}$ such that for $r\geq 1$ any graph of treewidth $\geq f(r)$ contains the $r\times r$ grid as a minor.
\end{theorem}

Robertson and Seymour~\cite{rs86} did not provide an explicit bound on
$f$, but proved it was bounded by a tower of exponentials.
The first explicit bounds on $f$ were given by Robertson, Seymour and
Thomas~\cite{rst94} who showed that $f(r)=2^{O(r^5)}$ suffices and
there are graphs which force $f(r)=\Omega(r^{2}\cdot \log r)$. The
question whether $f(r)$ can be shown to be bounded by a polynomial in $r$
was open for a long time until Chekuri and
Chuzhoy~\cite{julia-chanrda-poly-grid-minor} showed that
$f(r)=O(r^{98}\cdot \log^{O(1)} r)$ suffices. The current best bound
is $f(r)=O(r^{9}\cdot \log^{O(1)} r)$ due to
Chuzhoy and Tan~\cite{julia-soda-19}. Henceforth, for ease of presentation, we will use the weaker bound $f(r)=O(r^{10})$.
%\graham{``We state this'' sentence fragment left over}

%Theorem~\ref{thm:excluded-minor} gives a WIN-WIN approach for designing FPT algorithms as follows: either we have a small treewidth (and we can do dynamic programming), or otherwise the graph contains a large grid which gives allows for some structural insights.
The theory of bidimensionality~\cite{mth-bidimensionality-jacm,mth-bidimensionality-survey} exploits the idea that many problems can be solved efficiently via dynamic programming on graphs of bounded treewidth, and have large values on grid-like graphs.

\begin{definition}[\textbf{minor-bidimensional}]
A graph problem $\Pi$ is said to be $g(r)$-minor-bidimensional if
\begin{itemize}
  \item The value of $\Pi$ on the $r\times r$ grid is $\geq g(r)$
  \item $\Pi$ is closed under taking minors, i.e., the value of $\Pi$ does not increase under the operations of vertex deletions, edge deletions, edge contractions.
\end{itemize}
\label{dfn:g(r)-minor-bidimensional}
\end{definition}

Hence, we obtain a ``win-win'' approach for designing FPT algorithms for bidimensional problems as follows:
\begin{itemize}
  \item Either the graph has small treewidth and we can then use dynamic programming algorithms for bounded treewidth graphs; or
  \item The treewidth is large which implies that the graph contains a large grid as a minor. This implies that the solution size is large, since the parameter is minor-bidimensional.
\end{itemize}
Several natural graph parameters are known to be minor-bidimensional. For example, treewidth is $\Omega(r)$-minor-dimensional and Feedback  Vertex Set, Vertex Cover, Minimum Maximal Matching, Long Path, etc are $\Omega(r^2)$-minor-bidimensional.
To design parameterized streaming algorithms, we will replace the dynamic programming step for bounded treewidth graphs by simply storing all the edges of such graphs.

The following (folklore) lemma shows that bounded treewidth graphs cannot have too many edges.

\begin{lemma}
\label{lem:edge-upper-bound-in-bounded-tw}%{\cite[Exercise~7.14]{pc-book}}
Let $G=(V,E)$ be a graph on $n$ vertices. Then $|E(G)|\leq \tw(G)\cdot |V(G)|$
\end{lemma}
\begin{proof}
Let $\tw(G)=k$. We first show that there is a vertex $v\in G$ whose
degree in $G$ is at most $k$. Among all tree-decompositions of $G$ of
width $k$, let $(\T,B)$ be one which minimizes $|T|$. Since $T$ is a
tree, it has a leaf say $t$. Let $t'$ be the unique neighbor of $t$ in
$T$. Minimality of $|T|$ implies that $B_t\nsubseteq B_{t'}$, since
otherwise deleting $t$ (and the bag $B_t$) would still give a
tree-decomposition of $G$. Hence, there is a vertex $v\in B_t$ and
$v\notin B_{t'}$. This implies that all neighbors of $v$ in $G$ must
be in the bag $B_{t}$, i.e., $v$ has degree at most $|B_t|-1=k$. Now,
delete the vertex $v$. It follows from the definition of treewidth that $\tw(G-v)\leq \tw(G)=k$, and hence we can conclude that $G-v$ also has a vertex of degree at most $k$. Continuing this way, we obtain $|E(G)|\leq \tw(G)\cdot |V(G)|$.
\end{proof}

Note that cliques are a tight example (up to factor $2$) for the bound in Lemma~\ref{lem:edge-upper-bound-in-bounded-tw}.
%The theory of bidimensionality is based on a WIN-WIN approach: either the treewidth is large (and hence we have a large grid minor which implies large optimum), or the treewidth is small and we can do dynamic programming. Here, we will replace the second part of dynamic programming by simply storing all the linearly many edges of bounded treewidth graphs.
%By the result of Chuzhoy and Tan~\cite{julia-soda-19} we know that there is a constant $\tau>0$ such that every graph with treewidth $\geq \tau\cdot r^{10}$ has the $r\times r$ grid as a minor.
%Since $\Pi$ is $g(r)$-minor-dimensional, let $\beta>0$ be a constant such that value of $\Pi$ on the $r\times r$ grid is $\geq \beta\cdot r^2$.

%\subsubsection{General graphs}

%Let $h(r)=O(g^{19})\cdot \log^{O(1)} g$

\begin{lemma}
\label{lem:edge-characterization-general-graphs}
Let $\Pi$ be a $g(r)$-minor-dimensional problem. Any graph $G$ having more than $O((g^{-1}(k+1))^{10}\cdot |V(G)|)$ edges is a NO (resp. YES) instance of $k$-$\Pi$ if $\Pi$ is a minimization (resp. maximization) problem.
\end{lemma}
\begin{proof}
Suppose $G$ has more than $\tau (g^{-1}(k+1))^{10}\cdot n$ edges. By Lemma~\ref{lem:edge-upper-bound-in-bounded-tw}, it follows that $\tw(G)\geq \tau (g^{-1}(k+1))^{10}$. This implies $G$ has the $g^{-1}(k+1)\times g^{-1}(k+1)$ grid as a minor~\cite{julia-soda-19}. Since $\Pi$ is minor-dimensional, this implies that the value of $\Pi$ is at least $g(g^{-1}(k+1)) = k+1$, i.e., $G$ is a NO (resp. YES) instance of $k$-$\Pi$ if $\Pi$ is a minimization (resp. maximization) problem.
\end{proof}

Lemma~\ref{lem:edge-characterization-general-graphs} implies streaming algorithms for $\Pi$ in both insertion-only and insertion-deletion streams. First, we define a data structure that we need.

\begin{definition}[\textbf{$k$-sparse recovery algorithm}]
\label{dfn:k-sparse}
 A $k$-sparse recovery algorithm is a data structure which accepts
 insertions and deletions of elements from $[n]$ so that, if the current number of elements stored in it is at most $k$, then these can be recovered in full.
\end{definition}

Barkay et al.~\cite{ely-k-sparse} showed that a $k$-sparse recovery algorithm can be constructed deterministically using $\Ot(k)$ space.

\begin{theorem}
Let $M\geq 1$. Then we can check if a graph stream contains at most $M$ edges (and also store all these edges) using
\begin{itemize}
  \item $O(M)$ space in insertion-only streams
  \item $\Ot(M)$ space in insertion-deletion streams
\end{itemize}
\label{thm:storing-edges-meta-thm}
\end{theorem}
\begin{proof}
The algorithm in insertion-only streams simply stores all the edges. It also maintains a counter (using $O(\log n)$ bits) to count how many edges have been seen so far. If the counter exceeds $M$ then the graph has more than $M$ edges. Otherwise, we have stored the entire graph which uses $O(M)$ space since the number of edges is $\leq M$.

In insertion-deletion streams we also keep a counter (to count how many
edges are currently present) and also maintain an $M$-sparse recovery algorithm $\mathcal{X}$. At the end of the stream, if the counter exceeds $M$ then the graph stream has more than $M$ edges. Otherwise we recover the whole graph by extracting the $\leq M$ edges from $\mathcal{X}$. The counter can be implemented in $O(\log n)$ bits, and $\mathcal{X}$ can be implemented in $\Ot(M)$ space~\cite{ely-k-sparse}.
\end{proof}

Now we are ready to show the main theorem of this section: minor-bidimensional problems belong to the class \semips.

\begin{theorem}\label{thm:bidim-general}
(\textbf{minor-bidimensional problems are in \semips})
Let $\Pi$ be a $g(r)$-minor-dimensional problem. Then the $k$-$\Pi$ problem on graphs with $n$ vertices can be solved using
\begin{itemize}
  \item $O((g^{-1}(k+1))^{10}\cdot n)$ space in insertion-only streams
  \item $\Ot((g^{-1}(k+1))^{10}\cdot n)$ space in insertion-deletion streams
\end{itemize}
\end{theorem}
\begin{proof}
We invoke Theorem~\ref{thm:storing-edges-meta-thm} with $M=O((g^{-1}(k+1))^{10}\cdot n)$. By Lemma~\ref{lem:edge-characterization-general-graphs}, we know that if $G$ has more than $O(g^{-1}(k+1))^{10}\cdot n$ edges then $G$ is a NO (resp. YES) instance of $k$-$\Pi$ if $\Pi$ is a minimization (resp. maximization) problem. Hence, we use the algorithms from Theorem~\ref{thm:storing-edges-meta-thm} to check if $G$ has at most $M$ edges: if it has more edges then we say NO (resp. YES) if $\Pi$ is a minimization (resp. maximization) problem, and otherwise we store the entire graph.
%
%The algorithm in insertion-only streams simply stores all the edges. It also maintains a counter to count how many edges have been seen so far. If the counter exceeds $O(k^{10}\cdot |V(G)|)$ then we can say NO (by Lemma~\ref{lem:edge-characterization-general-graphs}). Otherwise, we have stored the entire graph which has at most $O(k^{10}\cdot |V(G)|)$ edges.
%
%In insertion-deletion streams we keep a counter (to count how many
%edges are currently present) and also maintain an $O(k^{10}\cdot n)$-sparse recovery algorithm $X$. At the end of the stream, if the counter exceeds $O(k^{10}\cdot |V(G)|)$ then we can say NO (by Lemma~\ref{lem:edge-characterization-general-graphs}). Otherwise we recover the whole graph by extracting the $O(k^{10}\cdot n)$ edges from $X$. The counter can be implemented in $O(\log n)$ bits, and $X$ can be implemented in $\Ot(k^{10}\cdot n)$ space~\cite{ely-k-sparse}.
\end{proof}

Theorem~\ref{thm:bidim-general} implies the following results for specific graph problems\footnote{We omit the simple proofs of why these problems satisfy the conditions of Definition~\ref{dfn:g(r)-minor-bidimensional}}:
\begin{itemize}
  \item Since Treewidth is $\Omega(r)$-minor-bidimensional, it follows
    that $k$-Treewidth has an $O(k^{10}\cdot n)$ space algorithm in insertion-only streams and $\Ot(k^{10}\cdot n)$ space algorithm in insertion-deletion streams.
  \item Since problems such as Long Path, Vertex Cover, Feedback Vertex Set, Minimum Maximal Matching, etc. are $\Omega(r^2)$-minor-bidimensional, it follows that their parameterized versions have $O(k^{5}\cdot n)$ space algorithm in insertion-only streams and $\Ot(k^{5}\cdot n)$ space algorithm in insertion-deletion streams.
\end{itemize}

In Section~\ref{sec:tight}, we design algorithms for some of the aforementioned problems with smaller storage. In particular, we design problem-specific structural lemmas (for example, Lemma~\ref{lem:k-path-extremal} and Lemma~\ref{lem:fvs-extremal}) to reduce the dependency of $k$ on the storage from $k^{O(1)}$ to $k$.

\begin{remark}
It is tempting to conjecture a lower bound complementing
Theorem~\ref{thm:bidim-general}:
for example, can we show that the bounds for minor-bidimensional problems are tight for \semips, i.e., they do not belong to \subps or even \fps? Unfortunately, we can rule out such a converse to Theorem~\ref{thm:bidim-general} via the two examples of Vertex Cover (VC) and Feedback Vertex Set (FVS) which are both $\Omega(r^2)$-minor-bidimensional. Chitnis et al.~\cite{rajesh-soda-15} showed that $k$-VC can be solved in $O(k^2)$ space and hence belongs to the class \fps. However, we show (Theorem~\ref{thm:k-fvs-lb}) that $k$-FVS cannot belong to \subps since it has a $\Omega(n\log n)$ bits lower bound for $k=0$.
\end{remark}

%\subsubsection{$H$-minor-free graphs}
%
%In this section, we give an improved algorithm for $k$-$\Pi$ on $H$-minor-free graphs.
%
%\begin{dfn}
%\label{dfn:minor-free}
%(\textbf{$H$-minor-free})
%Let $H,G$ be two undirected graphs. We say that $G$ is $H$-minor-free if $H$ if not a minor of $G$.
%\end{dfn}
%
%$H$-minor-free graphs are known to be sparse, i.e., there exists a constant $c_{H}$ such that any $H$-minor-free graph $G$ has at most $c_{H}\cdot |V(G)|$ edges. Hence, by storing linearly many edges we can simply store the whole graph. It is known~\cite{kostochka1982minimum,kostochka84,thomason01,thomason84} that $c_{H}=O(|H|\cdot \sqrt{\log |H|})$, and that this bound is tight when $H$ is a clique. Demaine and Hajiaghayi~\cite{mth-linearity} showed the following Excluded Grid Minor theorem for $H$-minor-free graphs:
%
%\begin{thm}
%For any fixed graph $H$, every $H$-minor-free graph of treewidth $\omega$ has the $\frac{\omega}{d_H}\times \frac{\omega}{d_H}$ grid as a minor, where $d_H$ is a constant depending on $|H|$
%\end{thm}
%
%It is known that $d_H=\Omega(\sqrt{|H|}\cdot \log |H|)$ otherwise it would violate the known lower bound for general graphs~\cite{rst94}.

\section{Lower Bounds for Parameterized Streaming Algorithms}
\label{sec:lb}

\subsection{Tight Problems for the classes \semips and \bps}
\label{sec:tight}

In this section we show that certain problems are tight for the classes \semips and \bps. All of the results hold for 1-pass in the insertion-only model. Our algorithms are deterministic, while the lower bounds also hold for randomized algorithms.

\subsubsection{Tight Problems for the class \semips}
\label{sec:tight-semi}
%\textbf{\underline{Tight Problems for the class \semips}}: 

We now show that some parameterized problems are tight for the class \semips, i.e.,
\begin{itemize}
  \item They belong to \semips, i.e., can be solved using $\Ot(g(k)\cdot n)$ bits for some function $g$.
  \item They do not belong to \subps, i.e., there is no algorithm which uses $\Ot(f(k)\cdot n^{1-\epsilon})$ bits for any function $f$ and any constant $1>\epsilon>0$. We do this by showing $\Omega(n\cdot \log n)$ bits lower bounds for these problems for constant values of $k$.
\end{itemize}
For each of the problems considered in this section, a lower bound of $\Omega(n)$ bits (for constant values of $k$) was shown by Chitnis et al.~\cite{rajesh-soda-16}. To obtain the improved lower bound of $\Omega(n\cdot \log n)$ bits for constant $k$, we will reduce from the \textsc{Perm} problem defined by Sun and Woodruff~\cite{woodruff-tight-bounds-insertion}.
\begin{center}
\noindent\framebox{\begin{minipage}{5in}
\textsc{Perm}\\
\emph{Input}: Alice has a permutation $\delta: [N]\rightarrow [N]$ which is represented as a bit string $B_{\delta}$ of length $N\log N$ by concatenating the images of $1,2,\ldots,N$ under $\delta$. Bob has an index $I\in [N\log N]$.\\
\emph{Goal}: Bob wants to find the $I$-th bit of $B_{\delta}$
%\emph{Question}: Does $G$ have a path of length at least $k$? (or alternatively, a path on at least $k+1$ vertices)
\end{minipage}}
\end{center}
%\vspace{-4mm}
Sun and Woodruff~\cite{woodruff-tight-bounds-insertion} showed that the one-way (randomized) communication complexity of \textsc{Perm} is $\Omega(N\cdot \log N)$. Using the \textsc{Perm} problem, we show $\Omega(n\cdot \log n)$ bit lower bounds for constant values of $k$ for various problem such as $k$-Path, $k$-Treewidth, $k$-Feedback-Vertex-Set, etc. We also show a matching upper bound for these problems: for each $k$, these problems can be solved using $O(kn\cdot \log n)$ words in insertion-only streams and $\Ot(kn\cdot \log n)$ words in insertion-deletion streams. 
%The proofs of these results are deferred to Appendix~\ref{app:tight-semi}. 
To the best of our knowledge, the only problems known previously to be tight for \semips were $k$-vertex-connectivity and $k$-edge-connectivity~\cite{crouch,woodruff-tight-bounds-insertion,eppstein-jacm}.

\begin{center}
\noindent\framebox{\begin{minipage}{5in}
\textbf{$k$-Path}\\
\emph{Input}: An undirected graph $G$ on $n$ nodes\\
\emph{Parameter}: $k$\\
\emph{Question}: Does $G$ have a path of length at least $k$? (or alternatively, a path on at least $k+1$ vertices)
\end{minipage}}
\end{center}

\begin{theorem}
The $k$-Path problem has a lower bound of $\Omega(n\cdot \log n)$ bits even for $k=5$.
\label{thm:k-path-lb}
\end{theorem}
\begin{proof}
Let $n=2N+2$. We start with an instance of \textsc{Perm} of size $N$. Alice has a permutation $\delta$ which she uses to build a perfect matching from $[N]$ to $[N]$ as follows: let $W=\{w_1, w_2, \ldots, w_N\}$ and $X=\{x_1, x_2, \ldots, x_N\}$ denote two sets of size $N$ each. Alice's edge set consists of a perfect matching built as follows: for each $i\in [N]$ there is an edge between $w_i$ and $x_{\delta(i)}$.
Suppose Bob has the index $I\in [N]$. This corresponds to the $\ell$-th bit of $\delta(j)$ for some $j\in [N]$ and $\ell\in [\log N]$. Bob adds two new vertices $v,y$ and adds edges using the index $I$ as follows:
\begin{itemize}
  \item Bob adds an edge between $v$ and $w_j$
  \item Let $S_{\ell}\subseteq X$ where $S_{\ell}=\{x_r\ :\ \ell{\text{-th}}\  \text{bit of}\ r\ \text{is}\ 0\}$. Bob adds edges from $y$ to each vertex of $S_{\ell}$.
\end{itemize}
Let the graph constructed this way be $G'$. It is easy to see that $G'$ has a path of length $5$ if and only $x_j\in S_{\ell}$, i.e., the $\ell$-th bit of $\delta(j)$ is zero. Hence, the lower bound of $\Omega(N\log N)$ of \textsc{Perm} translates to an $\Omega(n\log n)$ lower bound for $5$-Path.
\end{proof}

%Theorem~\ref{thm:k-path-lb} improves the lower bound of $\Omega(n)$ for $4$-path from Chitnis et al.~\cite{rajesh-soda-16}. As a corollary, we obtain
%
%\begin{crl}
%$k$-Path $\notin$\subps
%\end{crl}

\begin{lemma}
\label{lem:k-path-extremal}
Any graph on $n$ vertices with at least $nk$ edges has a path on $k+1$ vertices
\end{lemma}
\begin{proof}
Preprocess the graph to enforce that the minimum degree $\geq k$ by
iteratively deleting vertices of degree $<k$.
Then we have a graph $G'$ which has $n'$ vertices and $\geq n'k$ edges
whose min degree is $\geq k$.
Now consider an arbitrary path $P$ in this
graph $G'$, say $v_1-v_2-v_3-\ldots-v_r$.
At each intermediate vertex $v_j$, at most $j-1$ neighbors have been
visited, and so at least $k-j+1$ possibilities are open.
Hence, there is always a possible next step up to node $k+1$,
i.e. there is a path of length $k$.
\end{proof}

\begin{theorem}
\label{thm:algo-k-path}
The $k$-Path problem can be solved using
\begin{itemize}
  \item $O(k\cdot n)$ space in insertion-only streams
  \item $\Ot(k\cdot n)$ space in insertion-deletion streams
\end{itemize}
\end{theorem}
\begin{proof}
We invoke Theorem~\ref{thm:storing-edges-meta-thm} with $M=nk$. By Lemma~\ref{lem:k-path-extremal}, we know that if $G$ has more than $M$ edges then it has a $k$-Path. Hence, we use the algorithms from Theorem~\ref{thm:storing-edges-meta-thm} to check if $G$ has at most $M$ edges: if it has more edges then we say YES, and otherwise we store the entire graph.
%
%
%The algorithm in insertion-only streams simply stores all the edges. It also maintains a counter to count how many edges have been seen so far. If the counter exceeds $k\cdot |V(G)|$ then we can say YES (by Lemma~\ref{lem:k-path-extremal}). Otherwise, we have stored the entire graph which has at most $O(k\cdot |V(G)|)$ edges.
%
%In insertion-deletion streams we keep a counter (to count how many edges are present) and also maintain a $O(k\cdot n)$-sparse recovery algorithm $X$. At the end of the stream, if the counter exceeds $O(k\cdot |V(G)|)$ then we can say NO (by Lemma~\ref{lem:k-path-extremal}). Otherwise we recover the whole graph by extracting the $\leq O(k\cdot n)$ edges from $X$. The counter can be implemented in $O(\log n)$ bits, and $X$ can be implemented in $\Ot(k\cdot n)$ space~\cite{ely-k-sparse}.
\end{proof}

\begin{center}
\noindent\framebox{\begin{minipage}{5.00in}
\textbf{$k$-Treewidth}\\
\emph{Input}: An undirected graph $G$\\
\emph{Parameter}: $k$\\
\emph{Question}: Is the treewidth of $G$ at most $k$?
\end{minipage}}
\end{center}

\begin{theorem}~\cite[Theorem 7]{woodruff-tight-bounds-insertion}
The $k$-Treewidth problem has a lower bound of $\Omega(n\cdot \log n)$ bits even for $k=1$.
\label{thm:k-tw-lb}
\end{theorem}
\begin{proof}
Let $n=2N+1$. We start with an instance of \textsc{Perm} of size $N$. Alice has a permutation $\delta$ which she uses to build a perfect matching from $[N]$ to $[N]$ as follows: let $W=\{w_1, w_2, \ldots, w_N\}$ and $X=\{x_1, x_2, \ldots, x_N\}$ denote two sets of size $N$ each. Alice's edge set consists of a perfect matching built as follows: for each $i\in [N]$ there is an edge between $w_i$ and $x_{\delta(i)}$.
Suppose Bob has the index $I\in [N]$. This corresponds to the $\ell$-th bit of $\delta(j)$ for some $j\in [N]$ and $\ell\in [\log N]$. Bob adds a new vertex $v$ and adds edges using the index $I$ as follows:
\begin{itemize}
  \item Bob adds an edge between $v$ and $w_j$
  \item Let $S_{\ell}\subseteq X$ where $S_{\ell}=\{x_r\ :\ \ell\text{-th bit of}\ r\ \text{is}\ 0\}$. Bob adds edges from $v$ to each vertex of $S_{\ell}$.
\end{itemize}
Let the graph constructed this way be $G'$. It is easy to see that $G'$ has no cycles if and only $x_j\notin S_{\ell}$, i.e., the $\ell$-th bit of $\delta(j)$ is $1$. Recall that a graph has treewidth $1$ if and only if it has no cycles. Hence, the lower bound of $\Omega(N\log N)$ of \textsc{Perm} translates to a $O(n\log n)$ lower bound for $k$-Treewidth with $k=1$.
\end{proof}

\begin{theorem}
\label{thm:algo-k-tw}
The $k$-Treewidth problem can be solved using
\begin{itemize}
  \item $O(k\cdot n)$ space in insertion-only streams
  \item $\Ot(k\cdot n)$ space in insertion-deletion streams
\end{itemize}
\end{theorem}
\begin{proof}
We invoke Theorem~\ref{thm:storing-edges-meta-thm} with $M=nk$. By Lemma~\ref{lem:edge-upper-bound-in-bounded-tw}, we know that if $G$ has more than $M$ edges then $\tw(G)>k$. Hence, we use the algorithms from Theorem~\ref{thm:storing-edges-meta-thm} to check if $G$ has at most $M$ edges: if it has more edges then we say NO, and otherwise we store the entire graph.
%
%The algorithm in insertion-only streams simply stores all the edges. It also maintains a counter to count how many edges have been seen so far. If the counter exceeds $k\cdot |V(G)|$ then we can say NO (by Lemma~\ref{lem:edge-upper-bound-in-bounded-tw}). Otherwise, we have stored the entire graph which has at most $O(k\cdot |V(G)|)$ edges.
%
%In insertion-deletion streams we keep a counter (to count how many edges are present) and also maintain a $O(k\cdot n)$-sparse recovery algorithm $X$. At the end of the stream, if the counter exceeds $O(k\cdot |V(G)|)$ then we can say NO (by Lemma~\ref{lem:edge-upper-bound-in-bounded-tw}). Otherwise we recover the whole graph by extracting the $\leq O(k\cdot n)$ edges from $X$. The counter can be implemented in $O(\log n)$ bits, and $X$ can be implemented in $\Ot(k\cdot n)$ space~\cite{ely-k-sparse}.
\end{proof}

%\subsection{$k$-FVS}

\begin{center}
\noindent\framebox{\begin{minipage}{5.00in}
\textbf{$k$-FVS}\\
\emph{Input}: An undirected graph $G=(V,E)$\\
\emph{Parameter}: $k$\\
\emph{Question}: Does there exist a set $X\subseteq V$ such that $|X|\leq k$ and $G\setminus X$ has no cycles?
\end{minipage}}
\end{center}

\begin{theorem}
~\cite[Theorem 7]{woodruff-tight-bounds-insertion}
The $k$-FVS problem has a lower bound of $\Omega(n\cdot \log n)$ bits even for $k=0$.
\label{thm:k-fvs-lb}
\end{theorem}
\begin{proof}
%Let $n=2N+1$. We start with an instance of \textsc{Perm} of size $N$. Alice has a permutation $\delta$ which she uses to build a perfect matching from $[N]$ to $[N]$ as follows: let $W=\{w_1, w_2, \ldots, w_N\}$ and $X=\{x_1, x_2, \ldots, x_N\}$ denote two sets of size $N$ each. Alice's edge set consists of a perfect matching built as follows: for each $i\in [N]$ there is an edge between $w_i$ and $x_{\delta(i)}$.
%Suppose Bob has the index $I\in [N]$. This corresponds to the $\ell$-th bit of $\delta(j)$ for some $j\in [N]$ and $\ell\in [\log N]$. Bob adds a new vertex $v$ and adds edges using the index $I$ as follows:
%\begin{itemize}
%  \item Bob adds an edge between $v$ and $w_j$
%  \item Let $S_{\ell}\subseteq X$ where $S_{\ell}=\{x_r\ :\ \ell\text{-th bit of}\ r\ \text{is}\ 0\}$. Bob adds edges from $v$ to each vertex of $S_{\ell}$.
%\end{itemize}
%Let the graph constructed this way be $G'$. It is easy to see that $G'$ has no cycles if and only $x_j\notin S_{\ell}$, i.e., the $\ell$-th bit of $\delta(j)$ is $1$.
We use exactly the same reduction as in Theorem~\ref{thm:k-tw-lb}. Recall that graph has a FVS of size $0$ if and only if it has no cycles. Also a graph has treewidth $1$ if and only if it has no cycles. The proof of Theorem~\ref{thm:k-tw-lb} argues that $G'$ has no cycles if and only $x_j\notin S_{\ell}$, i.e., the $\ell$-th bit of $\delta(j)$ is $1$. Since $G'$ has $n=2n+1$ vertices, the lower bound of $\Omega(N\log N)$ of \textsc{Perm} translates to a $O(n\log n)$ lower bound for $k$-FVS with $k=0$.
\end{proof}

\begin{lemma}
If $G$ has a feedback vertex set of size $k$ then $|E(G)|\leq n(k+1)$
\label{lem:fvs-extremal}
\end{lemma}
\begin{proof}
Let $X$ be a feedback vertex of $G$ of size $k$. Then $G\setminus X$ is a forest and has at most $n-k-1$ edges. Each vertex of $X$ can have degree $\leq n-1$ in $G$. Hence, we have that $|E(G)|\leq (n-k-1)+k(n-1)\leq n+kn = n(k+1)$
\end{proof}

\begin{theorem}
\label{thm:algo-k-fvs}
The $k$-FVS problem can be solved using
\begin{itemize}
  \item $O(k\cdot n)$ space in insertion-only streams
  \item $\Ot(k\cdot n)$ space in insertion-deletion streams
\end{itemize}
\end{theorem}
\begin{proof}
We invoke Theorem~\ref{thm:storing-edges-meta-thm} with $M=n(k+1)$. By Lemma~\ref{lem:fvs-extremal}, we know that if $G$ has more than $M$ edges then $G$ cannot have a feedback vertex set of size $k$. Hence, we use the algorithms from Theorem~\ref{thm:storing-edges-meta-thm} to check if $G$ has at most $M$ edges: if it has more edges then we say NO, and otherwise we store the entire graph.
%
%The algorithm in insertion-only streams simply stores all the edges. It also maintains a counter to count how many edges have been seen so far. If the counter exceeds $n(k+1)$ then we can say NO (by Lemma~\ref{lem:edge-upper-bound-in-bounded-tw}). Otherwise, we have stored the entire graph which has at most $n(k+1)$ edges.
%
%In insertion-deletion streams we keep a counter (to count how many edges are present) and also maintain a $n(k+1)$-sparse recovery algorithm $X$. At the end of the stream, if the counter exceeds $n(k+1)$ then we can say NO (by Lemma~\ref{lem:edge-upper-bound-in-bounded-tw}). Otherwise we recover the whole graph by extracting the $n(k+1)$ edges from $X$. The counter can be implemented in $O(\log n)$ bits, and $X$ can be implemented in $\Ot(k\cdot n)$ space~\cite{ely-k-sparse}.
\end{proof}

\subsubsection{Tight Problems for the class \bps}
\label{sec:tight-bps}

%\textbf{\underline{Tight Problems for the class \bps}}: 

We now show that some parameterized problems are tight for the class \bps, i.e.,
\begin{itemize}
  \item They belong to \bps, i.e., can be solved using $O(n^2)$ bits. Indeed any graph problem can be solved by storing the entire adjacency matrix which requires $O(n^2)$ bits.
  \item They do not belong to \subps, i.e., there is no algorithm which uses $\Ot(f(k)\cdot n^{1+\epsilon})$ bits for any function $f$ and any $\epsilon\in (0,1)$. We do this by showing $\Omega(n^2)$ bits lower bounds for these problems for constant values of $k$ via reductions from the \textsc{Index} problem.
\end{itemize}
%We show cannot be solved in $f(k)\cdot
%n^{1+\epsilon}\cdot \log^{O(1)}$ bits for any function $f$ and any constant $\epsilon\in (0,1)$.
%We do this by showing $\Omega(n^2)$ bit lower bounds for these
%problems for constant values of $k$.
%Any graph problem can be solved with recourse to $O(n^2)$ bits of
%space, by storing the adjacency matrix of the graph, showing that this
%bound is tight.
%For our lower bounds, we reduce from the \textsc{Index} problem.
%\vspace{-4mm}
\begin{center}
\noindent\framebox{\begin{minipage}{5in}
\textbf{\textsc{Index}}\\
\emph{Input}: Alice has a string $B=b_1 b_2 \ldots b_N \in \{0,1\}^{N}$. Bob has an index $I\in [N]$\\
\emph{Goal}: Bob wants to find the value $b_I$
%\emph{Question}: Does there exist a set $X\subseteq V(G)$ with $|X|\leq k$ such that $G\setminus X$ has no edges?
\end{minipage}}
\end{center}

%\vspace{-4mm}
There is a $\Omega(N)$ lower bound on the (randomized) one-way communication complexity of \textsc{Index}~\cite{nisan-cc-book}. Via reduction from the \textsc{Index} problem, we are able to show $\Omega(n^2)$ bits for constant values of $k$ for several problems such as $k$-Dominating-Set and $k$-Girth. 

%The proofs of these reductions are deferred to Appendix~\ref{app:tight-bps}

%\begin{remark}
%We usually only design \FPT algorithms for \NP-hard problems. However, parameterized streaming algorithms make sense for all graph problems since we are only comparing ourselves against the naive choice of storing all the $O(n^2)$ bits of adjacency matrix. Hence, here we consider the $k$-Girth problem as an example of a polynomial time solvable problem.
%\end{remark}

\begin{center}
\noindent\framebox{\begin{minipage}{5in}
\textbf{$k$-Dominating Set}\\
\emph{Input}: An undirected graph $G=(V,E)$\\
\emph{Parameter}: $k$\\
\emph{Question}: Is there a set $S\subseteq V(G)$ of size $\leq k$ such that each $v\in V\setminus S$ has at least one neighbor in $S$?
\end{minipage}}
\end{center}

\begin{theorem}
The $k$-Dominating Set problem has a lower bound of $\Omega(n^2)$ bits for $1$-pass algorithms, even when $k=3$.
\label{thm:k-domset-lb}
\end{theorem}
\begin{proof}
Let $r=\sqrt{N}$. We start with an instance of \textsc{Index} where
Alice has a bit string $B\in \{0,1\}^N$.
Fix a canonical bijection $\phi:[N]\rightarrow  [r]\times [r]$. We now construct a graph with vertex set $Y={y_1, y_2, \ldots, y_r}$ and $W={w_1, w_2, \ldots, w_r}$. For each $I\in [N]$ we do the following:
\begin{itemize}
  \item If $B[I]=1$ then add the edge $y_{i'}-w_{i''}$ where $\phi_{I}=(i',i'')$
  \item If $B[I]=0$ then do not add any edge
\end{itemize}
Suppose Bob has the index $I^{*}\in [N]$. Let $\phi(I^*)=(\alpha,\beta)$ where $\alpha,\beta\in [r]$. Bob adds four new vertices $x_1, x_2, z_1$ and $z_2$. He also adds the following edges:
\begin{itemize}
  \item The edge $x_1 - x_2$
  \item The edge $z_1 - z_2$
  \item An edge from $x_1$ to each vertex of $Y\setminus y_{\alpha}$
  \item An edge from $z_1$ to each vertex of $W\setminus w_{\beta}$
\end{itemize}
Let the final constructed graph be $G$. A simple observation is that if a vertex has degree exactly $1$, then its unique neighbor can be assumed to be part of a minimum dominating set. We now show that $G$ has a dominating set of size $3$ if and only $B[I^*]=1$.

First suppose that $B[I^*]=1$, i.e., $y_{\alpha}-w_{\beta}$ forms an edge in $G$. Then we claim that $\{x_1, z_1, y_{\alpha}\}$ form a dominating set of size $3$. This is because $x_1$ dominates $x_2\cup (Y\setminus y_{\alpha})$, $z_1$ dominates $z_2\cup (W\setminus w_{\beta})$ and finally $y_{\alpha}$ dominates $w_{\beta}$.

Now suppose that $G$ has a dominating set $S$ of size $3$ but $y_{\alpha}-w_{\beta}\notin E(G)$. Since $x_2, z_2$ have degree $1$ we can assume that $\{x_1, z_1\}\subseteq S$. Let $S\setminus \{x_1, z_1\}=u$. We now consider different possibilities for $u$ and obtain a contradiction in each case:
\begin{itemize}
  \item $u=x_2$ or $u=z_2$: In this case the vertex $y_{\alpha}$ is not dominated by $S$
  \item $u\in (Y\setminus y_{\alpha})$: In this case the vertex $y_{\alpha}$ is not dominated by $S$
  \item $u\in (W\setminus w_{\beta})$: In this case the vertex $w_{\beta}$ is not dominated by $S$
  \item $u=y_{\alpha}$: In this case the vertex $w_{\beta}$ is not dominated by $S$ since $y_{\alpha}-w_{\beta}\notin E(G)$
  \item $u=w_{\beta}$: In this case the vertex $y_{\alpha}$ is not dominated by $S$ since $y_{\alpha}-w_{\beta}\notin E(G)$
\end{itemize}
Hence, the $3$-Dominating Set problem on graphs with $2r+4 = O(\sqrt{N})$ vertices can be used to solve instances of the \textsc{Index} problem of size $N$. Since \textsc{Index} has a lower bound of $\Omega(N)$, it follows that the $3$-Dominating Set problem on graphs of $n$ vertices has a lower bound of $\Omega(n^2)$ bits.
\end{proof}

%\subsection{$k$-Girth}

\begin{center}
\noindent\framebox{\begin{minipage}{5.00in}
\textbf{$k$-Girth}\\
\emph{Input}: An undirected graph $G$\\
\emph{Parameter}: $k$\\
\emph{Question}: Is the length of smallest cycle of $G$ equal to $k$?
\end{minipage}}
\end{center}

\begin{theorem}
The $k$-Girth problem has a lower bound of $\Omega(n^2)$ bits for $1$-pass algorithms, even when $k=3$.
\label{thm:k-girth-lb}
\end{theorem}
\begin{proof}
Let $r=\sqrt{N}$. We start with an instance of \textsc{Index} where
Alice has a bit string $B\in \{0,1\}^N$.
Fix a canonical bijection $\phi:[N]\rightarrow  [r]\times [r]$. We now construct a graph with vertex set $Y={y_1, y_2, \ldots, y_r}$ and $W={w_1, w_2, \ldots, w_r}$. For each $I\in [N]$ we do the following:
\begin{itemize}
  \item If $B[I]=1$ then add the edge $y_{i'}-w_{i''}$ where $\phi_{I}=(i',i'')$
  \item If $B[I]=0$ then do not add any edge
\end{itemize}
Suppose Bob has the index $I^{*}\in [N]$. Let $\phi(I^*)=(\alpha,\beta)$ where $\alpha,\beta\in [r]$. Bob adds a new vertex $z$ and adds the edges $z-y_{\alpha}$ and $z-w_{\beta}$. Let the final constructed graph be $G$. It is easy to see that $G\setminus z$ is bipartite, and hence has girth $\geq 4$ (we say the girth is $\infty$ if the graph has no cycle). The only edges incident on $z$ are to $y_{\alpha}$ and $w_{\beta}$. Hence, $G$ has a cycle of length 3 if and only if the edge $y_{\alpha}-w_{\beta}$ is present in $G$, i.e., $B[I^*]=1$. Hence, the $3$-Girth problem on graphs with $2r+1 = O(\sqrt{N})$ vertices can be used to solve instances of the \textsc{Index} problem of size $N$. Since \textsc{Index} has a lower bound of $\Omega(N)$, it follows that the $3$-Girth problem on graphs of $n$ vertices has a lower bound of $\Omega(n^2)$ bits.
\end{proof}

Super-linear lower bounds for multi-pass algorithms for $k$-Girth were shown in Feigenbaum et al.~\cite{andrew-sicomp-08}.

\begin{remark}
We usually only design \FPT algorithms for \NP-hard problems. However, parameterized streaming algorithms make sense for all graph problems since we are only comparing ourselves against the naive choice of storing all the $O(n^2)$ edges. Hence, here we consider the $k$-Girth problem as an example of a polynomial time solvable problem.
\end{remark}

\subsection{\mbox{Lower bound for approximating size of minimum Dominating Set} on graphs of bounded arboricity}
\label{sec:dom-set-lb}

\begin{theorem}
\label{thm:dom-set-lb}
Let $\beta\geq 1$ be any constant. Then any streaming algorithm which $\frac{\beta}{32}$-approximates the size of a min dominating set on graphs of arboricity $\beta+2$ requires $\Omega(n)$ space.
%, i.e., essentially storing the whole graph.
\end{theorem}

Note that Theorem~\ref{thm:dom-set-lb} shows that the naive algorithm which stores all the $O(n\beta)$ edges of an $\beta$-arboriticy graph is essentially optimal. Our lower bound holds even for randomized algorithms (required to have success probability $\geq 3/4$) and also under the vertex arrival model, i.e., we see at once all edges incident on a vertex.
We \textbf{(very) closely} follow the outline from~\cite[Theorem 4]{set-cover-stoc-16} who used this approach for showing that any $\alpha$-approximation for estimating size of a minimum dominating set in general graphs\footnote{They actually talk about Set Cover} requires $\tilde{\Omega}(\frac{n^2}{\alpha^2})$ space. Because we are restricted to bounded arboriticy graphs, we cannot just use their reduction as a black-box but need to adapt it carefully for our purposes.

Let $V(G)=[n+1]$, and $\FC_{\beta}$ be the collection of all subsets of $[n]$ with cardinality $\beta$. Let \dsest be the problem of $\frac{\beta}{32}$-approximately estimating the size of a minimum dominating set on graphs of arboricity $\beta+2$. Consider the following distribution $\diste$ for \dsest.

\textbox{Distribution $\diste$: \textnormal{A hard input distribution for \dsest.}}{
	\medskip %\\
        %\textbf{Notation.} Let $\FC$ be the collection of all subsets of $[n]$ with         cardinality $\beta$.
\begin{itemize}
\item \textbf{Alice.} The input of Alice is a collection of $n$ sets $\mathcal{S}' = \{S'_{1}, S'_{2}, \ldots, S'_{n}\}$ where for each $i \in [n]$ we have that $S'_{i}=\{i\}\cup S_i$ with $S_i$ being a set chosen independently and uniformly at random from $\FC_{\beta}$.
\item \textbf{Bob.} Pick $\theta \in \set{0,1}$ and $\istar \in [n]$ independently and
  uniformly at random; the input of Bob is a single set $T$ defined as follows.
  \begin{itemize}
  \item If $\theta = 0$, then $\bar{T}=[n]\setminus T$ is a set of size $\beta/8$ chosen uniformly
    at random from $S_\istar$. %\footnote{    Since $\alpha = o(\sqrt{n/\log{n}})$ and $m = \poly(n)$, the size of    $\bar{T}$ is strictly smaller than the size of $S_\istar$.}
  \item If $\theta = 1$, then $\bar{T}=[n]\setminus T$ is a set of size $\beta/8$ chosen uniformly
    at random from $[n] \setminus S_\istar$.
  \end{itemize}
\end{itemize}
}
Recall that $\opt(\mathcal{S}',T)$ denotes the size of the minimum \emph{dominating set} of the
graph $G$ whose edge set is given by $N[i]=\{i\}\cup S_{i}$ for each $i\in [n]$ and $N[n+1]=\{n+1\}\cup T$. It is easy to see that $G$ has arboricity $\leq (\beta+2)$ since it has $(n+1)$ vertices and $\leq (\beta+1)n+ (1+n-\frac{\beta}{8})$ edges. We first establish the following lemma regarding the parameter $\theta$ and $\opt(\mathcal{S}',T)$
in the distribution $\diste$.

\begin{lemma}%$[\star]$
\label{lem:theta-random}
	Let $\alpha=\frac{\beta}{32}$. Then, for $(\mathcal{S}',T) \sim \diste$ we have
	\begin{enumerate}%[(i)]
		\item %\label{part:theta0}
                    $\PR{\opt(\mathcal{S}',T) = 2 \mid \theta=0} = 1$.
		\item %\label{part:theta1}
                    $\PR{\opt(\mathcal{S}',T) > 2\alpha \mid \theta = 1} = 1-o(1)$.
	\end{enumerate}
\end{lemma}
\begin{proof}
  The first claim is immediate since by construction, when $\theta = 0$ we have that $T
  \union S'_\istar = [n+1]$ and hence $\{n+1, \istar\}$ forms a dominating set of size $2$.

  We now prove the second claim, i.e., when $\theta=1$. The vertex $(n+1)$ dominates all vertices in the set $T\cup \{n+1\}$. It remains to dominate vertices of $\bar{T}=[n]\setminus T$. Since $i\in S'_{i}$ for each $i\in [n]$ it follows that the set $\{j\ :\ j\in \bar{T}\}\cup \{n+1\}$ forms a dominating set of size $1+\frac{\beta}{8}$ for $G$.   Fix a collection $\widehat{\mathcal{S}'}$ of $2\alpha$ sets in $\mathcal{S}' \setminus \set{S'_\istar}$, and let $\widehat{\mathcal{S}'} = \{S'_{\mu_{1}}, S'_{\mu_2}, \ldots, S'_{\mu_{2\alpha}}\}$. Let $T_0 = \bar{T}\setminus \{\mu_1, \mu_2, \ldots, \mu_{2\alpha}\}$, and note that $|T_0|= |\bar{T}|-2(\frac{\beta}{32}) = \frac{\beta}{16}$. Hence, we have that $\widehat{\mathcal{S}} = \{S_{\mu_{1}}, S_{\mu_2}, \ldots, S_{\mu_{2\alpha}}\}$ has to cover $T_0$, where $S_{\mu_{j}} = S'_{\mu_{j}}$ for each $1\leq j\leq 2\alpha$ and the sets $\{S_{\mu_{1}}, S_{\mu_2}, \ldots, S_{\mu_{2\alpha}}\}$ are chosen independent of $T_0$ (according to the distribution $\diste$). We first analyze the probability that $\widehat{\mathcal{S}}$ covers $T_0$ and then take union bound over
  all choices of $2\alpha$ sets from $\mathcal{S}' \setminus \set{S'_\istar}$.

  %Note that according to the distribution $\diste$, the sets in   $ \{S_{\mu_{1}}, S_{\mu_2}, \ldots, S_{\mu_{2\alpha}}\}$ are drawn independent of $\bar{T}$.

  Fix any choice of $T_0$; for each
  element $k \in T_0$, and for each set $S_j \in \widehat{\mathcal{S}}$, define an indicator random
  variable $\bX^{j}_k \in \set{0,1}$, where $\bX^{j}_k=1$ iff $k \in S_j$.  Let $\bX :=
  \sum_{j}\sum_{k} \bX^{j}_k$ and notice that:

  \[
	\Ex[\bX] = \sum_{j}\sum_{k} \Ex[\bX^{j}_k] = (2\alpha) \cdot (\frac{\beta}{16}) \cdot (\frac{\beta}{n}) = \frac{\alpha \beta^{2}}{8n}
	\]
	We have,
	\[
		\PR{\text{$\widehat{\mathcal{S}}$ covers $T_0$}} \leq \PR{\bX \geq \frac{\beta}{16}} = \PR{\bX \geq \frac{ n}{2\alpha\beta} \cdot \Ex[\bX]}
	\]
	It is easy to verify that the $\bX^{j}_k$ variables are negatively
        correlated. Hence, applying the extended Chernoff bound\footnote{Let $\bX_1, \bX_2,\ldots, \bX_r$ be a sequence of negatively correlated Boolean random variables, and let $\bX=\sum_{i=1}^{r} \bX_i$. Then $\Pr(|\bX-\Ex[\bX]|\geq \epsilon\cdot \Ex[\bX]) \leq 3\cdot \text{exp}(\frac{-\epsilon^{2}\Ex[\bX]}{3})$} due to Panconesi and Srinivasan~\cite{DBLP:journals/siamcomp/PanconesiS97} we get
    \[
		\PR{\bX \geq \frac{n}{2\alpha \beta} \cdot \Ex[\bX]} \leq 3\exp\paren{\frac{-\epsilon^{2} \Ex[\bX]}{3}}\ \text{where}\ 1+\epsilon=\frac{ n}{2\alpha\beta}
	\]

        Finally, by union bound,
	\begin{align*}
		\Pr(\opt(\mathcal{S}',T) \leq 2\alpha) &\leq \PR{\exists~\Salpha \text{ covers } T_0} \leq {n \choose 2\alpha} \cdot 3\exp\paren{\frac{-\epsilon^{2} \Ex[\bX]}{3}}
\\
		&\leq \exp\paren{2\alpha\cdot\log{n}}\cdot 3\exp\paren{\frac{-\epsilon^{2} \Ex[\bX]}{3}}
%= o(1)
	\end{align*}

Since $\alpha=\frac{\beta}{32}$, one can easily check that $\exp\paren{\frac{-\epsilon^{2} \Ex[\bX]}{3}} \leq \exp\paren{-3\alpha\cdot \log n} $ and hence we have
\begin{align*}
		\Pr(\opt(\mathcal{S}',T) \leq 2\alpha) &\leq \PR{\exists~\Salpha \text{ covers } \bar{T}} \leq {n \choose 2\alpha} \cdot 3\exp\paren{\frac{-\epsilon^{2} \Ex[\bX]}{3}}
\\
		&\leq \exp\paren{2\alpha\cdot\log{n}}\cdot 3\exp\paren{\frac{-\epsilon^{2} \Ex[\bX]}{3}}\\
        &\leq \exp\paren{2\alpha\cdot\log{n}}\cdot 3\exp\paren{-3\alpha\cdot \log n}\\
        &= o(1)
	\end{align*}

\end{proof}

Observe that distribution $\diste$ is not a product distribution due to the correlation
between the input given to Alice and Bob. However, we can express the distribution as a
convex combination of a relatively small set of product distributions. To do so, we need the following definition.  For
integers $k,t$ and $n$, a collection $P$ of $t$ subsets of $[n]$ is called a \emph{random
  $(k,t)$-partition} if the $t$ sets in $P$ are constructed as follows: Pick $k$ elements
from $[n]$, denoted by $S$, uniformly at random, and partition $S$ randomly into $t$ sets
of equal size.  We refer to each set in $P$ as a \emph{block}.

\textbox{An alternative definition of the distribution \diste.}{
\begin{enumerate}
	\item[] \textbf{Parameters:} $~~~~~~~~~~k = 2\beta~~~~~~~~~~p = \frac{\beta}{8}~~~~~~~~~~t = 16$
	\item For any $i \in [n]$, let $P_i$ be a random $(k,t)$-partition in $[n]$ (chosen independently).
	\item The input to Alice is $\mathcal{S}' = \paren{S'_1,\ldots,S'_n}$, where for each $i$ we have $S'_i = \{i\}\cup S_i$ and $S_i$ is created by picking $t/2$ blocks from $P_i$ uniformly at random.
	\item The input to Bob is a set $T$ where $\bar{T}$ is created by first picking an
          $\istar \in [n]$ uniformly at random, and then picking a block from $P_\istar$
          uniformly at random.
\end{enumerate}
}
To see that the two formulations of the distribution $\diste$ are indeed equivalent, notice that $(i)$ the input given to Alice in the new formulation is a collection of sets of size $\beta$ chosen
independently and uniformly at random (by the independence of $P_i$'s), and $(ii)$ the complement of the set given to Bob is a set of size $\frac{\beta}{8}$ which, for $\istar \in_R [n]$, with probability half, is chosen uniformly at
random from $S_\istar$, and with probability half, is chosen from $[n] \setminus S_{\istar}$ (by the randomness in the choice of each block in $P_\istar$).

\subsubsection{The Lower Bound for the Distribution $\diste$}
\label{app:lb-distribution}

\newcommand{\bTbar}{\ensuremath{\mathbf{\bar{T}}}\xspace}

\textbf{Notation}: First we set up some notation to be used throughout this section: we use bold face letters to represent random variables. For any random variable $\bX$,
$\supp{\bX}$ denotes its support set. We define $\card{\bX}: = \log{\card{\supp{\bX}}}$.  For any $k$-dimensional tuple $X = (X_1,\ldots,X_k)$ and any $i \in [k]$, we
define $X^{<i}:=(X_1,\ldots,X_{i-1})$, and $X^{-i}:=(X_1,\ldots,X_{i-1},X_{i+1},\ldots,X_{k})$.  The notation ``$X\in_R U$'' indicates that $X$ is chosen uniformly at random from a set $U$. We denote the \emph{Shannon Entropy} of a random variable $\bA$ by
$H(\bA)$ and the \emph{mutual information} of two random variables $\bA$ and $\bB$ by
$I(\bA;\bB) = H(\bA) - H(\bA \mid \bB) = H(\bB) - H(\bB \mid \bA)$. If the distribution
$\dist$ of the random variables is not clear from the context, we use $H_\dist(\bA)$
(resp. $I_{\dist}(\bA;\bB)$). We use $H_2$ to denote the binary entropy function where for any real number $0
< \delta < 1$, $H_2(\delta) = \delta\log{\frac{1}{\delta}} +
(1-\delta)\log{\frac{1}{1-\delta}}$.

We are now ready to show our desired lower bound: fix any $\delta$-error protocol $\protds$ (set $\delta = 1/4$) for \dsest on the
distribution $\diste$. Recall that $\bprotds$ denotes the random variable for the concatenation of the message of Alice with the
public randomness used in the protocol $\protds$. We further use $\bPP := (\bP_1,\ldots,\bP_t)$ to denote the random partitions $(P_1,\ldots,P_{t})$, $\bI$ for the choice of the special index $\istar$, and $\bT$
for the parameter $\theta \in \set{0,1}$, whereby $\theta = 0$ if and only if $\bar{T} \subseteq S_\istar$.

%%Note that under the distribution $\diste$, the inputs to Alice and Bob are
%%random variables. In the following, we denote these random variables as follows.  $\bS
%%:= (\bS_1,\ldots,\bS_m)$ is for the set $\SA = \set{S_1,\ldots,S_m}$ of Alice and $\bTbar$
%%is for the set $\bar{T}$ of Bob. We further use $\bP := (\bP_1,\ldots,\bP_{t})$ for the partitions
%%$(P_1,\ldots,P_{t})$. $\bI$ is for the choice of the \emph{special index} $\istar$, and $\bT$ is for the parameter $\theta$,
%%i.e., $\theta = 0$ iff $\bar{T} \subseteq S_\istar$. Finally, recall that $\bprotds$ denotes the concatenation of the message of Alice with the
%%public randomness used in the protocol.

We make the following simple observations about the distribution $\diste$. The proofs are straightforward.

\begin{remark}\label{rem:dist} In the distribution $\diste$,
	\begin{enumerate}
		\item \label{p1} The random variables $\bSS$, $\bPP$, and $\bprotds(\bSS)$ are all
                  independent of the random variable $\bI$.
		\item \label{p2} For any $i \in [m]$, conditioned on $\bP_i = P$, and $\bI = i$, the
                  random variables $\bS_i$ and $\bTbar$ are independent of each
                  other. Moreover, $\supp{\bS_i}$ and $\supp{\bTbar}$ contain,
                  respectively, ${t \choose \frac{t}{2}}$ and $t$ elements and both
                  $\bS_i$ and $\bTbar$ are uniform over their support.
		\item \label{p3} For any $i \in [m]$, the random variable $\bS_i$ is
                  independent of both $\bSS^{-i}$ and $\bPP^{-i}$.
%%                  , where $\bS^{-i}$ and
%%                  $\bP^{-i}$ denote respectively the random variables $\paren{\bS_1,
%%                   \ldots, \bS_{i-1}, \bS_{i+1}, \ldots \bS_m}$ and \\$\paren{\bP_1,
%%                   \ldots, \bP_{i-1}, \bP_{i+1}, \ldots \bP_m}$.
%%	
	\end{enumerate}
\end{remark}

We will show our claimed lower bound on the space complexity of any streaming algorithm for \dsest through the well-known connection to one-way communication complexity~\cite{nisan-cc-book}.

\begin{definition}
  The \emph{communication cost} of a protocol $\prot$ for a problem $P$ on an input
  distribution $\dist$, denoted by $\norm{\prot}$, is the worst-case size of the message
  sent from Alice to Bob in the protocol $\prot$, when the inputs are chosen from the distribution
  $\dist$. \newline The \emph{communication complexity} $\CC{P}{\dist}{\delta}$ of a
  problem $P$ with respect to a distribution $\dist$ is the minimum communication cost of
  a $\delta$-error protocol $\prot$ over $\dist$.
\end{definition}

However, since the information complexity is a lower bound on the communication complexity~\cite{amit-yao}, we will instead bound the information complexity of $\protds$ (set $\delta = 1/4$) for \dsest on the
distribution $\diste$.

\begin{definition}
  Consider an input distribution $\dist$ and a protocol $\prot$ (for some problem
  $P$). Let $\bX$ be the random variable for the input of Alice drawn from $\dist$, and let
  $\bprot := \bprot(\bX)$ be the random variable denoting the message sent from Alice to
  Bob \emph{concatenated} with the public randomness $\bR$ used by $\prot$.  The
  information cost $\ICost{\prot}{\dist}$ of a one-way protocol $\prot$ with respect to
  $\dist$ is $I_\dist(\bprot;\bX)$. \newline The \emph{information
    complexity} $\IC{P}{\dist}{\delta}$ of $P$ with respect to a distribution $\dist$ is
  the minimum $\ICost{\prot}{\dist}$ taken over all one-way $\delta$-error protocols
  $\prot$ for $P$ over $\dist$.
\end{definition}

Our goal now is to lower bound $\ICost{\protds}{\diste}$, which then also gives a desired lower bound on $\norm{\protds}$. We start by simplifying the expression for $\ICost{\protds}{\diste}$.

\begin{lemma}%$[\star]$
\label{lem:info-cost-bound}
	$\ICost{\protds}{\diste} \geq \sum_{i=1}^{n} I(\bprotds ; \bS_i \mid \bP_i)$
\end{lemma}
\begin{proof}
	We have,
	\begin{align*}
		\ICost{\protds}{\diste}&= I(\bprotds; \bSS) \geq I(\bprotds ; \bSS \mid \bPP)
	\end{align*}
	where the inequality holds since $(i)$ $H(\bprotds) \geq H(\bprotds \mid \bPP)$ and $(ii)$ $H(\bprotds \mid \bSS) = H(\bprotds \mid \bSS,\bPP)$ as $\bprotds$ is independent of $\bPP$ conditioned on $\bSS$.
	We now bound the conditional mutual information term in the above equation.
	\begin{align*}
		 I(\bprotds ; \bSS \mid \bPP) &= \sum_{i=1}^{m} I(\bS_i ; \bprotds \mid \bPP,\bSS^{<i}) \tag{the chain rule for the mutual information} \\
		&= \sum_{i=1}^{m} H(\bS_i \mid \bPP,\bSS^{<i}) - H(\bS_i \mid \bprotds,\bPP,\bSS^{<i}) \\
		&\geq \sum_{i=1}^{m} H(\bS_i \mid \bP_i) - H(\bS_i \mid \bprotds,  \bP_i) \\
                 & = \sum_{i=1}^{m} I(\bS_i ; \bprotds \mid \bP_i)
	\end{align*}
	The inequality holds since:
	\begin{enumerate}[(i)]
	\item $H(\bS_i \mid \bP_i) = H(\bS_i \mid \bP_i,\bPP^{-i}, \bSS^{<i}) = H(\bS_i \mid \bPP,\bSS^{<i})$
	because conditioned on $\bP_i$, $\bS_i$ is independent of $\bPP^{-i}$ and $\bSS^{<i}$ (Remark~\ref{rem:dist}-(\ref{p3})), hence the equality since conditioning reduces the entropy.
	\item
        $H(\bS_i \mid \bprotds,\bP_i) \geq H(\bS_i \mid \bprotds, \bP_i, \bPP^{-i},\bSS^{<i}) =
        H(\bS_i \mid \bprotds, \bPP,\bSS^{<i})$ since conditioning reduces the entropy.
        %, i.e.,         Claim~\ref{clm:it-facts}-(\ref{part:cond-reduce}).
	\end{enumerate}	
\end{proof}

Equipped with Lemma~\ref{lem:info-cost-bound}, we only need to bound  $\sum_{i \in [n]} I(\bprotds ; \bS_i \mid \bP_i)$. Note that,
\begin{align}
	\sum_{i=1}^{n} I(\bprotds ; \bS_i \mid \bP_i) = \sum_{i=1}^{n} H(\bS_i \mid \bP_i) - \sum_{i=1}^{n} H(\bS_i \mid \bprotds , \bP_i) \label{eq:minfo-bound}
\end{align}
Furthermore, for each $i \in [n]$, $\card{\supp{\bS_i \mid \bP_i}} = {t \choose \frac{t}{2}}$ and $\bS_i$ is uniform over its support (Remark~\ref{rem:dist}-(\ref{p2})); hence we have
%Claim~\ref{clm:it-facts}-(\ref{part:uniform}),
\begin{align}
	\sum_{i=1}^{n} H(\bS_i \mid  \bP_i) =  \sum_{i=1}^{n} \card{\bS_i \mid  \bP_i} =  \sum_{i=1}^{n} \log{\card{\supp{\bS_i \mid  \bP_i}}} = \sum_{i=1}^{n} \log{t \choose \frac{t}{2}}
%= n\cdot \log{\paren{2^{t - \Theta(\log{t})}}}
= 13.64n
%n\cdot t - \Theta(n\log{t})
\label{eq:ent-bound}
\end{align}
 since $t=16$. Consequently, we only need to bound $\sum_{i=1}^{n} H(\bS_i \mid \bprotds, \bP_i)$. In order to do so, we show that $\protds$ can be used to estimate the value of
 the parameter $\theta$, and hence we only need to establish a lower bound for the problem of estimating $\theta$.

\begin{lemma}\label{lem:sc-to-theta}
  Any $\delta$-error protocol $\protds$ over the distribution $\diste$ can be used to
  determine the value of $\theta$ with error probability $\delta + o(1)$.
\end{lemma}
\begin{proof}
  Alice sends the message $\protds(\SS)$ as before. Using this message, Bob can compute an
  $\alpha$-estimation of the dominating set problem using $\protds(\SS)$ and his input.
  If the estimation is less than $2\alpha$, we output $\theta = 0$ and otherwise we output $\theta=1$. The
  bound on the error probability follows from Lemma~\ref{lem:theta-random}.
\end{proof}

%Before continuing, we make the following remark which would be useful in the next section:
%We assume that in \dsest over the distribution $\diste$, Bob is additionally provided with the special index $\istar$.
\begin{remark}\label{rem:known-i}
  We assume that in \dsest over the distribution $\diste$, Bob is additionally provided with the special index $\istar$.
\end{remark}
Note that this assumption can only make our lower bound stronger since Bob can always ignore this information and solve the original \dsest.

Let $\gamma$ be the function that estimates $\theta$ used in Lemma~\ref{lem:sc-to-theta}; the input to $\gamma$ is the message given from Alice, the public coins used by the players,
the set $\bar{T}$, and (by Remark~\ref{rem:known-i}) the special index $\istar$. We have,
\begin{align*}
  \Pr(\gamma(\bprotds,\bTbar,\bI) \neq \bT) \leq \delta+o(1)
\end{align*}

Hence, by Fano's inequality\footnote{For any binary random variable $\bB$ and any (possibly randomized) function f that
predicts $\bB$ based on $\bA$, if $\Pr($f(\bA)$ \neq \bB) = \delta$, then $H(\bB \mid \bA)
\leq H_2(\delta)$.}
%(Claim~\ref{clm:fano}),
\begin{align}
	H_2(\delta+o(1)) &\geq H(\bT \mid \bprotds,\bTbar, \bI) \notag \\
	&= \Ex_{i \sim \bI} \Bracket{H(\bT \mid \bprotds,\bTbar, \bI = i)} \notag \\
	&= \frac{1}{n} \sum_{i=1}^{m} H(\bT \mid \bprotds,\bTbar, \bI = i) \label{eq:i-sum}
\end{align}

We now show that each term above is lower bounded by $H(\bS_i \mid \bprotds,\bP_i) / t$ and hence we obtain the desired upper bound on $H(\bS_i \mid \bprotds,\bP_i)$ in Equation~(\ref{eq:minfo-bound}).

\begin{lemma}%$[\star]$
\label{lem:theta-to-s}
  For any $i \in [n]$, $H(\bT \mid \bprotds , \bTbar , \bI = i) \geq  H(\bS_i \mid \bprotds,\bP_i) / t$.
\end{lemma}
\begin{proof}
	We have,
	\begin{align*}
		H(\bT \mid \bprotds, \bTbar , \bI = i)&\geq H(\bT \mid \bprotds,\bTbar ,  \bP_i , \bI = i) \tag{conditioning on random variables reduces entropy} \\
		&= \Ex_{P \sim \bP_i \mid \bI = i} \Bracket{H(\bT \mid \bprotds,\bTbar,  \bP_i = P, \bI = i)}
	\end{align*}
	For brevity, let $E$ denote the event $(\bP_i = P , \bI = i)$.  We can write the above equation as,
	\begin{align*}
		 H(\bT \mid \bprotds,\bTbar ,  \bP_i , \bI = i) = \Ex_{P \sim \bP_i \mid \bI = i} \Ex_{\bar{T} \sim \bTbar \mid E} \Bracket{H(\bT \mid \bprotds,\bTbar = \bar{T},  E)}
	\end{align*}
	Note that by Remark~\ref{rem:dist}-(\ref{p2}), conditioned on the event $E$,
        $\bar{T}$ is chosen to be one of the blocks of $P = (B_1,\ldots,B_t)$ uniformly at
        random. Hence,
	\begin{align*}
          H(\bT \mid \bprotds,\bTbar ,  \bP_i , \bI = i) &= \Ex_{P \sim \bP_i \mid \bI = i} \Bracket{ \sum_{j=1}^{t}  \frac{H(\bT \mid \bprotds,\bTbar = B_j, E )}{t}}
	\end{align*}

        Define a random variable $\bX:= (\bX_1,\ldots,\bX_{t})$, where each $\bX_j \in \set{0,1}$ and
        $\bX_j = 1$ if and only if $\bS_i$ contains the block $B_j$. Note that conditioned on $E$,
        $\bX$ uniquely determines the set $\bS_i$. Moreover, notice that
        conditioned on $\bTbar = B_j$ and $E$, $\bT = 0$ if and only if $\bX_j = 1$. Hence,
	\begin{align*}
		 H(\bT \mid \bprotds,\bTbar ,  \bP_i , \bI = i) &= \Ex_{P \sim \bP_i \mid \bI = i} \Bracket{\sum_{j=1}^{t} \frac{H(\bX_j \mid \bprotds,\bTbar = B_j, E )}{t}} %\tag{$\bT = 1- \bX_j$ conditioned on $E$, and $\bTbar = B_j$}
	\end{align*}
	Now notice that $\bX_j$ is independent of the event $\bTbar = B_j$ since $\bS_i$
        is chosen independent of $\bTbar$ conditioned on $E$
        (Remark~\ref{rem:dist}-(\ref{p2})). Similarly, since $\bprotds$ is only a
        function of $\bS$ and $\bS$ is independent of $\bTbar$ conditioned on $E$, $\bprotds$
        is also independent of the event $\bTbar = B_j$. Consequently, we can ``drop'' the conditioning
        on $\bTbar = B_j$,
	\begin{align*}
	H(\bT \mid \bprotds,\bTbar ,  \bP_i , \bI = i) &= \Ex_{P \sim \bP_i \mid \bI = i} \Bracket{ \sum_{j=1}^{t} \frac{H(\bX_j \mid \bprotds, E )}{t}}  \\
		 &\geq \Ex_{P \sim \bP_i \mid \bI = i} \Bracket{ \frac{H(\bX \mid \bprotds, E )}{t}}  \tag{sub-additivity of the entropy} \\
		 &= \Ex_{P \sim \bP_i \mid \bI = i} \Bracket{\frac{H(\bS_i \mid \bprotds, E)}{t}} \tag{$\bS_i$ and $\bX$ uniquely define each other conditioned on $E$} \\
		 &= \Ex_{P \sim \bP_i \mid \bI = i} \Bracket{\frac{H(\bS_i \mid \bprotds, \bP_i = P, \bI = i)}{t}} \tag{$E$ is defined as $(\bP_i = P, \bI = i)$} \\
		 &= \frac{H(\bS_i \mid \bprotds,\bP_i,\bI = i)}{t}
	\end{align*}
	Finally, by Remark~\ref{rem:dist}-(\ref{p1}), $\bS_i$, $\bprotds$, and $\bP_i$ are all independent of the event $\bI = i$, and hence we have that $H(\bS_i \mid  \bprotds,\bP_i ,\bI = i) = H(\bS_i \mid \bprotds,\bP_i)$, which concludes the proof.
\end{proof}

By plugging in the bound from Lemma~\ref{lem:theta-to-s} in Equation~(\ref{eq:i-sum}) we have,
\begin{align*}
	\sum_{i=1}^{n} H(\bS_i \mid \bprotds , \bP_i) \leq H_2(\delta+o(1)) \cdot (nt) = 0.812\times (16n)= 12.992n
\end{align*}
 since $\delta=1/4$ and $t=16$. Finally, by plugging in this bound together with the bound from Equation~(\ref{eq:ent-bound}) in Equation~(\ref{eq:minfo-bound}), we get,
\begin{align*}
		\sum_{i=1}^{n} I(\bprotds ; \bS_i \mid \bP_i) &\geq 0.64n
%nt - \Theta(n\log{t}) - H_2(\delta+o(1)) \cdot (nt) \\
%		&= \Paren{1-H_2(\delta+o(1))} \cdot (nt) - \Theta(n \log{t})
\end{align*}

%Recall that $t =8$ and $\delta < 1/2$, hence $H_2(\delta + o(1)) = 1-\eps$ for some constant $\eps$ bounded away from $0$.
By Lemma~\ref{lem:info-cost-bound},

\begin{align*}
	\IC{\dsest}{\diste}{1/4} = \min_{\protds}\Paren{\ICost{\protds}{\diste}} = \Omega(n)
\end{align*}
To conclude, since the information complexity is a lower bound on the communication complexity~\cite{amit-yao}, we obtain a lower bound of $\Omega(n)$ for \dsest over the distribution $\diste$. This completes the proof of Theorem~\ref{thm:dom-set-lb}.

\subsection{Streaming Lower Bounds Inspired by Kernelization Lower Bounds}
\label{sec:lb-from-parameterized}

%\todo{Are these next two proofs formal enough?}

%\subsection{Using W-hardness}

Streaming algorithms and kernelization are two (somewhat related) compression models. In kernelization, we have access to the whole input but our computation power is limited to polynomial time whereas in streaming algorithms we don't have access to the whole graph (have to pay for whatever we store) but have unbounded computation power on whatever part of the input we have stored.

A folklore result states that a (decidable) problem is FPT if and only if it has a kernel. Once the fixed-parameter tractability for a problem is established, the next natural goals are to reduce the running time of the FPT algorithm and reduce the size of the kernel. In the last decade, several frameworks have been developed to show (conditional) lower bounds on the size of kernels~\cite{DBLP:journals/jcss/BodlaenderDFH09,dell-and,dell,DBLP:conf/focs/Drucker12,fortnow-santhanam}. Inspired by these frameworks, we define a class of problems, which we call as AND-compatible and OR-compatible, and show (unconditionally) that none of these problems belong to the class \subps.

\begin{definition}
We say that a graph problem $\Pi$ is \emph{AND-compatible} if there exists a constant $k$ such that
\begin{itemize}
  \item for every $n\in \mathbb{N}$ there exists a graph $G_{\emph{YES}}$ of size $n$ such $\Pi(G_{\emph{YES}},k)$ is a YES instance
  \item for every $n\in \mathbb{N}$ there exists a graph $G_{\emph{NO}}$ of size $n$ such $\Pi(G_{\emph{NO}},k)$ is a NO instance
  \item for every $t\in \mathbb{N}$ we have that $\Pi\Big(\uplus_{i=1}^{t} G_i, k\Big) = \wedge_{i=1}^{t} \Pi(G_i,k)$ where $G=\uplus_{i=1}^{t} G_i$ denotes the union of the vertex-disjoint graphs $G_1, G_2, \ldots, G_t$
\end{itemize}
\label{defn:and-compatible}
\end{definition}

Examples of \emph{AND-compatible} graph problems are $k$-Treewidth, $k$-Girth, $k$-Pathwidth, $k$-Coloring, etc.

\begin{definition}
We say that a graph problem $\Pi$ is \emph{OR-compatible} if there exists a constant $k$ such that
\begin{itemize}
  \item for every $n\in \mathbb{N}$ there exists a graph $G_{\text{YES}}$ of size $n$ such $\Pi(G_{\text{YES}},k)$ is a YES instance
  \item for every $n\in \mathbb{N}$ there exists a graph $G_{\text{NO}}$ of size $n$ such $\Pi(G_{\text{NO}},k)$ is a NO instance
  \item for every $t\in \mathbb{N}$ we have that $\Pi(\uplus_{i=1}^{t} G_i, k) = \vee_{i=1}^{t} \Pi(G_i,k)$ where $G=\uplus_{i=1}^{t} G_i$ denotes the union of the vertex-disjoint graphs $G_1, G_2, \ldots, G_t$
\end{itemize}
\label{defn:or-compatible}
\end{definition}

A general example of an \emph{OR-compatible} graph problem is the subgraph isomorphism problem parameterized by size of smaller graph: given a graph $G$ of size $n$ and a smaller graph $H$ of size $k$, does $G$ have a subgraph isomorphic to $H$? Special cases of this problem are $k$-Path, $k$-Clique, $k$-Cycle, etc.

\begin{theorem}
If $\Pi$ is an \emph{AND-compatible} or an \emph{OR-compatible} graph problem then $\Pi\notin \subps$
\label{thm:and-compatible-not-in-subps}
\end{theorem}
\begin{proof}
Let $\Pi$ be an \emph{AND-compatible} graph problem, and $G=\uplus_{i=1}^{t} G_i$ for some $t\in \mathbb{N}$. We claim that any streaming algorithm ALG for $\Pi$ must use $t$ bits. Intuitively, we need at least one bit to check that each of the instances $(G_i, k)$ is a YES instance of $\Pi$ (for all $1\leq i\leq t$). Consider a set of $t$ graphs $\mathcal{G}=\{G_1, G_2, \ldots, G_t\}$: note that we don't fix any of these graphs yet. For every subset $X\subseteq [t]$ we define the instance $(G_X,k)$ of $\Pi$ where $G_{X} = \uplus_{j\in J} G_j$. Suppose that ALG uses less than $t$ bits. Then by pigeonhole principle, there are two subsets $I,I'$ of $[t]$ such that ALG has the same answer on $(G_{I},k)$ and $(G_{I'},k)$. Since $I\neq I'$ (without loss of generality) there exists $i^*$ such that $i^* \in I\setminus I'$. This is where we now fix each of the graphs in $\mathcal{G}$ to arrive at a contradiction: consider the input where $G_{i}=G_{\text{YES}}$ for all $(I\cup I')\setminus i^*$ and $G_{i^*}=G_{\text{NO}}$. Then, it follows that $(G_{I},k)$ is a NO instance but $(G_{I'},k)$ is a YES instance.

Suppose that $\Pi\in \subps$, i.e., there is an algorithm for $\Pi$ which uses $f(k)\cdot N^{1-\epsilon}\cdot \log^{O(1)} N$ bits (for some $1>\epsilon>0$) on a graph $G$ of size $N$ to decide whether $(G,k)$ is a YES or NO instance. Let $G=\uplus_{i=1}^{t} G_i$ where $|G_i|=n$ for each $i\in [t]$. Then $|G|=N=nt$. By the previous paragraph, we have that
$$ f(k)\cdot (nt)^{1-\epsilon}\cdot \log^{O(1)} (nt) \geq t \Rightarrow f(k)\cdot n^{1-\epsilon}\cdot \log^{O(1)} (nt) \geq t^{\epsilon}$$

Choosing $t=n^{\frac{2-\epsilon}{\epsilon}}$ we have that $ f(k)\cdot \log^{O(1)} n^{1+(\frac{2-\epsilon}{\epsilon})} \geq n $, which is a contradiction for large enough $n$ (since $k$ and $\epsilon$ are constants).

We now prove the lower bound for AND-compatible problems. Recall that De Morgan's law states that $\neg (\vee_{i} P_i) = \wedge_{i} (\neg P_i) $. Hence, if  $\Pi$ is an \emph{OR-compatible} graph problem then the complement\footnote{By complement, we mean that $\bar{\Pi}(G,k)$ is YES if and only if $\Pi(G,k)$ is NO} problem $\bar{\Pi}$ is an \emph{AND-compatible} graph problem, and hence the lower bound follows from the previous paragraph.
%This theorem then follows from Theorem~\ref{thm:and-compatible-not-in-subps}.
\end{proof}

%\begin{theorem}
%If $\Pi$ is an \emph{OR-compatible} graph problem then $\Pi\notin \subps$
%\label{thm:or-compatible-not-in-subps}
%\end{theorem}
%\begin{proof}
%De Morgan's law states that $\neg (\vee_{i} P_i) = \wedge_{i} (\neg P_i) $. Hence, if  $\Pi$ is an \emph{OR-compatible} graph problem then the complement\footnote{By complement, we mean that $\bar{\Pi}(G,k)$ is YES if and only if $\Pi(G,k)$ is NO} problem $\bar{\Pi}$ is an \emph{AND-compatible} graph problem.
%This theorem then follows from Theorem~\ref{thm:and-compatible-not-in-subps}.
%\end{proof}

\begin{remark}
Note that throughout this paper we have considered the model where we allow unbounded computation at each edge update, and also at the end of the stream. However, if we consider a \textbf{restricted} model of allowing only polynomial (in input size $n$) computation at each edge update and also at end of the stream, then it is easy to see that existing (conditional) lower bounds from the parameterized algorithms and kernelization setting translate easily to this restricted model. For example, the following two lower bounds for parameterized streaming algorithms follow immediately in the restricted (polytime computation) model:
\begin{itemize}
  \item Let $X$ be a graph problem that is $W[i]$-hard parameterized by $k$ (for some $i\geq 1$). Then (in the polytime computation model) $X\notin \fps$ unless $\FPT=W[i]$.
  \item Let $X$ be a graph problem that is known to not have a polynomial kernel unless $\NP\subseteq \mathrm{co}\NP/\mathrm{poly}$. Then (in the polytime computation model) $X$ does not have a parameterized streaming algorithm which uses $k^{O(1)}\cdot \log^{O(1)} n$ bits, unless $\NP\subseteq \mathrm{co}\NP/\mathrm{poly}$.
\end{itemize}
\end{remark}

\subsection{$\Omega((N/d)^d)$ bits lower bound for $d$-SAT}
\label{app:lb-sat}

Finally, in this section, we show that for any $d\geq 2$, any streaming algorithm for $d$-SAT (in the clause arrival model) must essentially store all the clauses (and hence fits into the ``brute-force'' streaming setting). This is the only non-graph-theoretic result in this paper, and may be viewed as a ``streaming analogue'' of the Exponential Time Hypothesis. We fix the notation as follows: there are $N$ variables and $M$ clauses. The variable set is fixed, and the clauses arrive one-by-one.

\begin{theorem}
\label{thm:lb-sat}
Any streaming algorithm for $d$-SAT requires storage of $\Omega((N/d)^d)$ bits, where $N$ is the number of variables
\end{theorem}
\begin{proof}
%In this section, we show how to obtain a $\Omega((N/d)^d)$ bit lower bound for storage of any streaming algorithm for $d$-SAT, where $N$ is the number of variables.
  For simplicity, we show the result for $2$-SAT; the generalization
  to $d$-SAT for other $d > 2$ is simple.
%  but it is easy to from the proof that it generalizes for all $d\geq 2$.
Let $n=(N/2)^2$. We reduce from the \textsc{Index} problem. Let Alice have a string $B=b_1 b_2 \ldots b_n \in \{0,1\}^n$. We now
map $B$ to an instance $\phi_B$ of $2$-SAT defined over $N$
variables. The $N$ variables are partitioned into $d=2$ sets $X,Y$ of
$N/2$ variables each. Fix a canonical mapping $\psi:
[(N/2)^2]\rightarrow [N/2]^2$.
For each index $L\in [(N/2)^2]$, we add the following clauses depending on the value of $b_L$:
\begin{itemize}
  \item If $b_L=0$, then add the clause $(x_i \vee y_j)$ where $\psi(L)=(i,j)$.
  \item If $b_L=1$, then add the clause $(\bar{x_i} \vee y_j)$ where $\psi(L)=(i,j)$.
\end{itemize}
Observe that the sub-instance constructed so far is trivially
satisfiable, by setting all $y\in Y$ to true. Suppose Bob has the index $L^{*}\in [n]$. To solve the instance of \textsc{Index}, we need to retrieve the
value of the bit $b_{L^*}$.
Let $\psi(L^*)=(i^*, j^*)$. We add two new clauses as follows:
\begin{itemize}
  \item Add the clause $(\bar{y}_{j^*})$\hspace{1mm} \footnote{If we insist that
    all clauses should have cardinality exactly 2, then we can simply
    create a new ``dummy'' variable $z$, and add the clauses $(y_{j^*}
    \vee z), (y_{j^*} \vee \bar{z})$ to achieve the same effect.}
  \item Add the clause $(x_{i^*} \vee y_{j^*})$
\end{itemize}
This completes the construction of the $2$-SAT instance $\phi_B$. Now we claim that $\phi_B$ is satisfiable if and only if $b_{L^*}=0$. Consider a clause of the form $(x \vee y)$ of $\phi_B$:
\begin{itemize}
%  \item If $z\neq z_k$, then we can set z true to satisfy it
  \item If $y\neq y_{j^*}$ , we can set $y$ to be true and satisfy this clause.
  \item If $y= y_{j^*}$ but $x\neq x_{i^*}$, then we can satisfy this clause by setting $x=1$ (or $0$, if $x$ appears in complemented form). This is the only time we need to set $x$, and each such $x$ appears in at most one such clause, so there is no clash\footnote{Note that we do not have to know in what form $x$ appears in the input, as our question is just whether the
instance is satisfiable, not to provide a satisfying assignment.}.
\end{itemize}
This only leaves the clause on the variables $x_{i^*}$ and $y_{j^*}$.
We must set $y_{j^*}=0$ to satisfy the clause $\bar{y}_{j^*}$.
If $b_{L^*}=1$ then we have both the clauses $({x_{i^*}} \vee y_{j^*})$
and $(\bar{x}_{i^*} \vee y_{j^*})$, and hence the instance $\phi_B$ is not
satisfiable.
However, if $b_{L^*}=0$ then we only have the clause $(x_{i^*} \vee y_{j^*})$, and hence the instance $\phi_B$ is satisfiable by setting $x_{i^*}=1$.
Hence, the lower bound of $\Omega(n)=\Omega((N/2)^2)$ translates from \textsc{Index} to $2$-SAT.
%
%If the clause in the original input was $(x_i \vee y_j \vee z_k)$, then we can
%set $x_i$ true, and the instance is satisfied.
%
%But if the clause was $(\bar{x_i} \vee y_j \vee z_k)$, the system is not
%satisfiable, as we also need to satisfy $(x_i \vee y_j \vee z_k)$, and we have forced $y_j$ and $z_k$ to be false.
%
%This allows us to distinguish the two cases, and hence recover any bit
%from the original input string.
%By the communication complexity of the INDEX problem, the (one way
%randomized) communication complexity of 3SAT is Omega(N^3).
\end{proof}

Note that the naive algorithm for $d$-SAT which stores all the clauses in memory requires $\Ot(\binom{N}{d})=\Ot(d\cdot N^d)$ bits, and therefore Theorem~\ref{thm:lb-sat} shows that $d$-SAT is hard from a space perspective (essentially have to store all the clauses) for all
$d\ge 2$, whereas there is a transition from \Pe to \NP-complete for
the time cost when going from 2-SAT to 3-SAT.

\section{Conclusions and Open Problems}

In this paper, we initiate a systematic study of graph problems from the paradigm of parameterized streaming algorithms. We define space complexity classes of \fps, \subps, \semips, \supps and \bps, and then obtain tight classifications for several well-studied graph problems such as Longest Path, Feedback Vertex Set, Girth, Treewidth, etc. into these classes. Our parameterized streaming algorithms use techniques of bidimensionality, iterative compression and branching from the FPT world. In addition to showing lower bounds for some parameterized streaming problems via communication complexity, we also show how (conditional) lower bounds for kernels and W-hard problems translate to lower bounds for parameterized streaming algorithms.

Our work leaves open several concrete questions. We list some of them below:
\begin{itemize}
  \item The streaming algorithm (Algorithm~\ref{alg:streaming-branching}) for $k$-VC (on insertion-only streams) from
    Section~\ref{subsec:branching} has an optimal storage of $O(k\log n)$ bits but requires $2^{k}$ passes. Can we reduce the number of passes to $\poly(k)$, or instead show that we need passes which are superpolynomial in $k$ if we restrict space usage to $O(k\log n)$ bits? The only known lower bound for such algorithms is $(k/\log n)$ passes (see Theorem~\ref{thm:abboud-lb}).

\item As in the FPT setting, a natural problem to attack using iterative compression in the streaming setting would be the $k$-OCT problem. It is known that $0$-OCT, i..e, checking if a given graph is bipartite, in the $1$-pass model has an upper bound of $O(n\log n)$ bits~\cite{andrew-semi-streaming} and a lower bound of $\Omega(n\log n)$ bits~\cite{woodruff-tight-bounds-insertion}. For $k\geq 1$, can we design a $g(k)$-pass algorithm for $k$-OCT which uses $\Ot(f(k)\cdot n)$ bits for some functions $f$ and $g$, maybe using iterative compression? To the best of our knowledge, such an algorithm is not known even for $1$-OCT.

%  \item For $k\geq 1$ can we design algorithms which use $f(k)\cdot n\cdot \log^{O(1)} n$ bits and $g(k)$ passes for the $k$-OCT problem (for some functions $f,g$)? The technique of iterative compression seems like a natural tool to use here.
 \end{itemize}

% \thispagestyle{empty}
%
% \end{titlepage}

%\input{intro2}
% \input{intro}
%\input{scheme}

% \newpage

\bibliography{papers}

%\newpage

%\appendix

%\input{appendix}

%\input{bidir}
%\input{appendix}
%\input{alg}
%\input{hard} % don't forget to uncomment for submission!!
%\input{inapprox}
%\input{subsumed-dsn-inapprox}
%\input{open}

\end{document}